\documentclass[sigconf]{aamas} 
\usepackage{booktabs} 
\usepackage{graphicx,array}
\usepackage{balance}
\usepackage{color,soul,xspace}
\usepackage{times}
\usepackage{hyperref}
\usepackage[ruled,vlined,linesnumbered]{algorithm2e}
\SetKwInput{KwInput}{Input}
\SetKwInput{KwOutput}{Output}
\usepackage{comment}
\usepackage{epsfig}
\usepackage{subfig}
\usepackage{mathrsfs,amsmath,mdwlist,enumitem}
\DeclareMathOperator*{\argmax}{arg\,max}
\newtheorem{thm}{Theorem}
\newtheorem{cor}[thm]{Corollary}
\newtheorem{observation}{Observation}

\setlist{nolistsep,leftmargin=*}
\usepackage{booktabs}

\usepackage{flushend}
\acmDOI{doi}  
\acmISBN{}  
\acmYear{2020}  
\copyrightyear{2020}  
\acmPrice{}  

\begin{document}
\title{K-Core Minimization: A Game Theoretic Approach}


\author{Sourav Medya}
\affiliation{%
  \institution{Northwestern University}
}
\email{sourav.medya@kellogg.northwestern.edu}
\author{Tiyani Ma}
\affiliation{%
  \institution{University of California Los Angeles}
}
\email{tianyima@cs.ucla.edu}
\author{Arlei Silva}
\affiliation{%
  \institution{University of California Santa Barbara}
}
\email{arlei@cs.ucsb.edu}
\author{Ambuj Singh}
\affiliation{%
  \institution{University of California Santa Barbara}
}
\email{ambuj@cs.ucsb.edu}  

\begin{abstract} 


$K$-cores are maximal induced subgraphs where all vertices have degree at least $k$. These dense patterns have applications in community detection, network visualization and protein function prediction.  However, $k$-cores can be quite unstable to network modifications, which motivates the question: \textit{How resilient is the k-core structure of a network, such as the Web or Facebook, to edge deletions?} We investigate this question from an algorithmic perspective. More specifically, we study \textit{the problem of computing a small set of edges for which the removal minimizes the $k$-core structure of a network}.

This paper provides a comprehensive characterization of the hardness of the $k$-core minimization problem (KCM), including innaproximability and fixed-parameter intractability. Motivated by such a challenge in terms of algorithm design, we propose a novel algorithm inspired by Shapley value---a cooperative game-theoretic concept--- that is able to leverage the strong interdependencies in the effects of edge removals in the search space. As computing Shapley values is also NP-hard, we efficiently approximate them using a randomized algorithm with probabilistic guarantees. Our experiments, using several real datasets, show that the proposed algorithm outperforms competing solutions in terms of $k$-core minimization while being able to handle large graphs. Moreover, we illustrate how KCM can be applied in the analysis of the $k$-core resilience of networks.

\end{abstract}


\maketitle


\section{Introduction}
\label{sec:intro}

$K$-cores play an important role in revealing the higher-order organization of networks. 
A $k$-core \cite{seidman1983network} is a maximal induced subgraph where all  vertices have internal degree of at least $k$. These cohesive subgraphs have been applied to model users' engagement and viral marketing in social networks \cite{bhawalkar2015preventing,kitsak2010identification}. Other applications include anomaly detection \cite{shin2016corescope}, community discovery \cite{peng2014accelerating}, protein function prediction \cite{you2013prediction},  and visualization  \cite{alvarez2006large,carmi2007model}. However, the $k$-core structure can be quite unstable under network modification. For instance, removing only a few edges from the graph might lead to the collapse of its core structure. 
This motivates the $k$-core minimization problem:
\textit{ Given a graph G and constant k, find a small set of $b$ edges for which the removal minimizes the size of the k-core structure \cite{zhu2018k}.} 


We motivate $k$-core minimization using the following applications: (1) \textit{Monitoring:} Given an infrastructure or technological network, which edges should be monitored for attacks \cite{xiangyu2013identification,Laishram2018}? (2) \textit{Defense:} Which communication channels should be blocked in a terrorist network in order to destabilize its activities \cite{pedahzur2006changing,perliger2011social}? and (3) \textit{Design:} How to prevent unraveling in a social or biological network by strengthening connections between nodes \cite{bhawalkar2015preventing,morone2018k}? 

Consider a specific application of $k$-cores to online social networks (OSNs). OSN users tend to perform activities (e.g., joining a group, playing a game) if enough of their friends do the same \cite{burke2009feed}. Thus, strengthening critical links between users is key to the long-term popularity, and even survival, of the network \cite{Farzan2011}. This scenario can be modeled using $k$-cores. Initially, everyone is engaged in the $k$-core.  Removal of a few links (e.g., unfriending, unfollowing) might not only cause a couple of users to leave the network but produce a mass exodus due to cascading effects. This process can help us to understand the decline and death of OSNs such as Friendster \cite{garcia2013social}.

\begin{figure}[t]
    \centering
     \subfloat[Initial $G$]{\includegraphics[width=0.24\columnwidth]{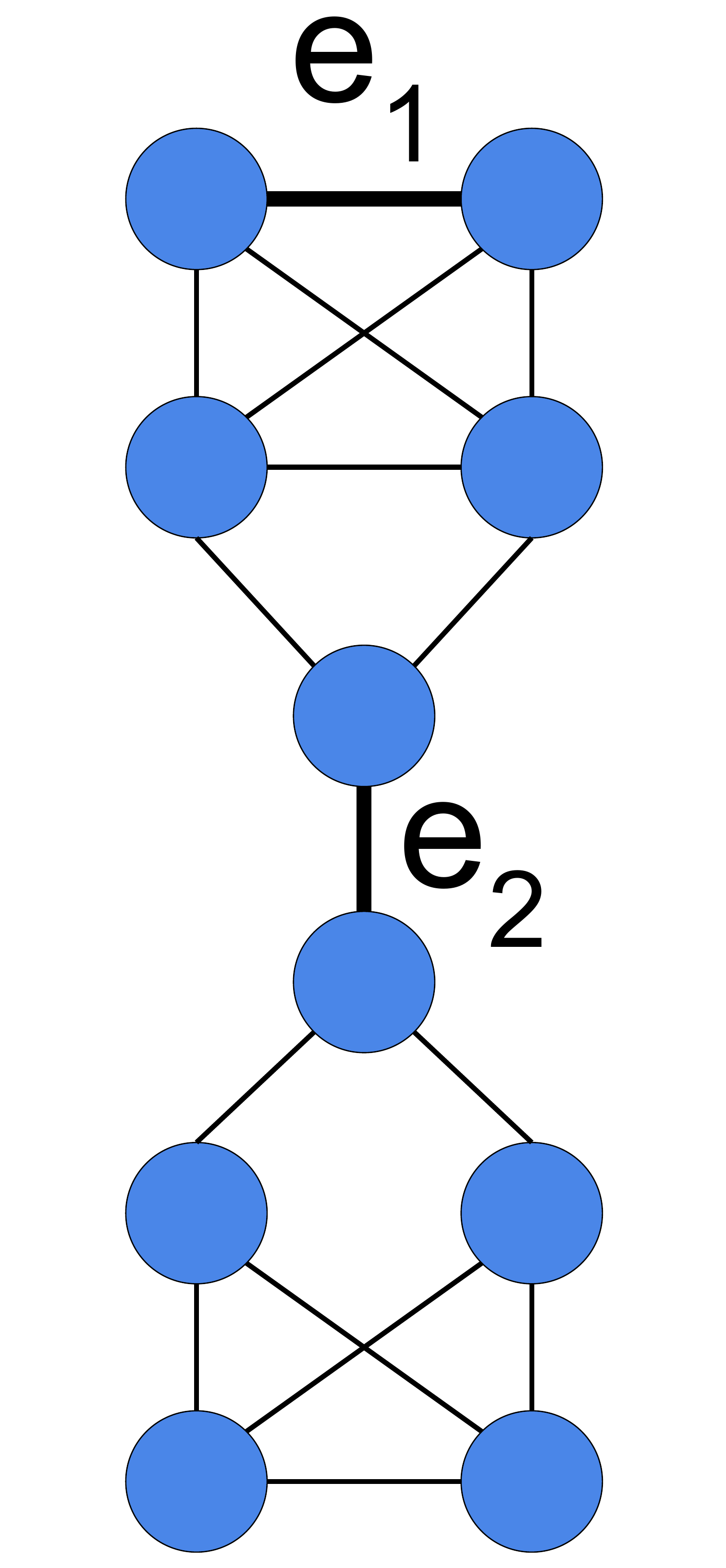}\label{fig:intro_1}}
     \hspace{3em}
      \subfloat[Modification $G'$]{\includegraphics[width=0.24\columnwidth]{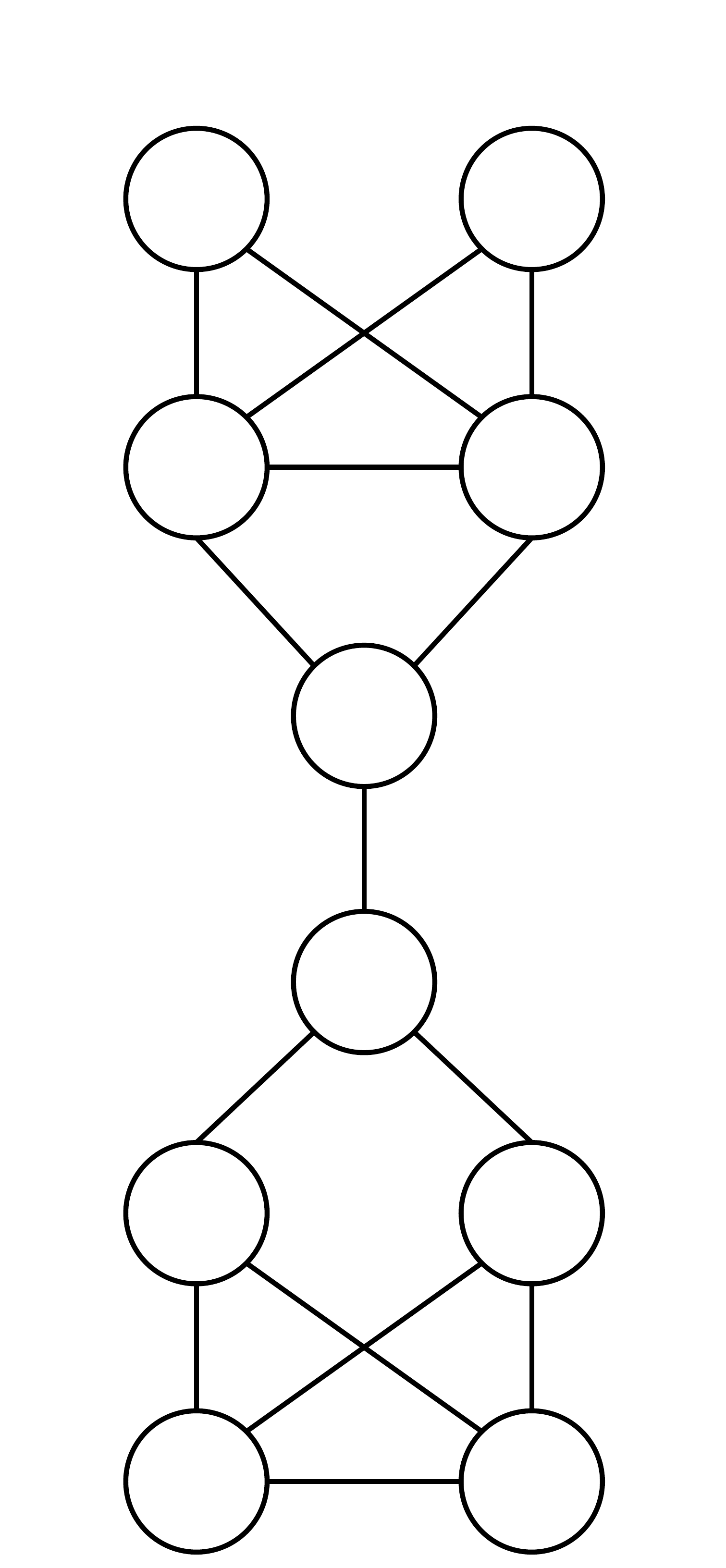}\label{fig:intro_2}}
      \hspace{3em}
      \subfloat[Modification $G''$]{\includegraphics[width=0.24\columnwidth]{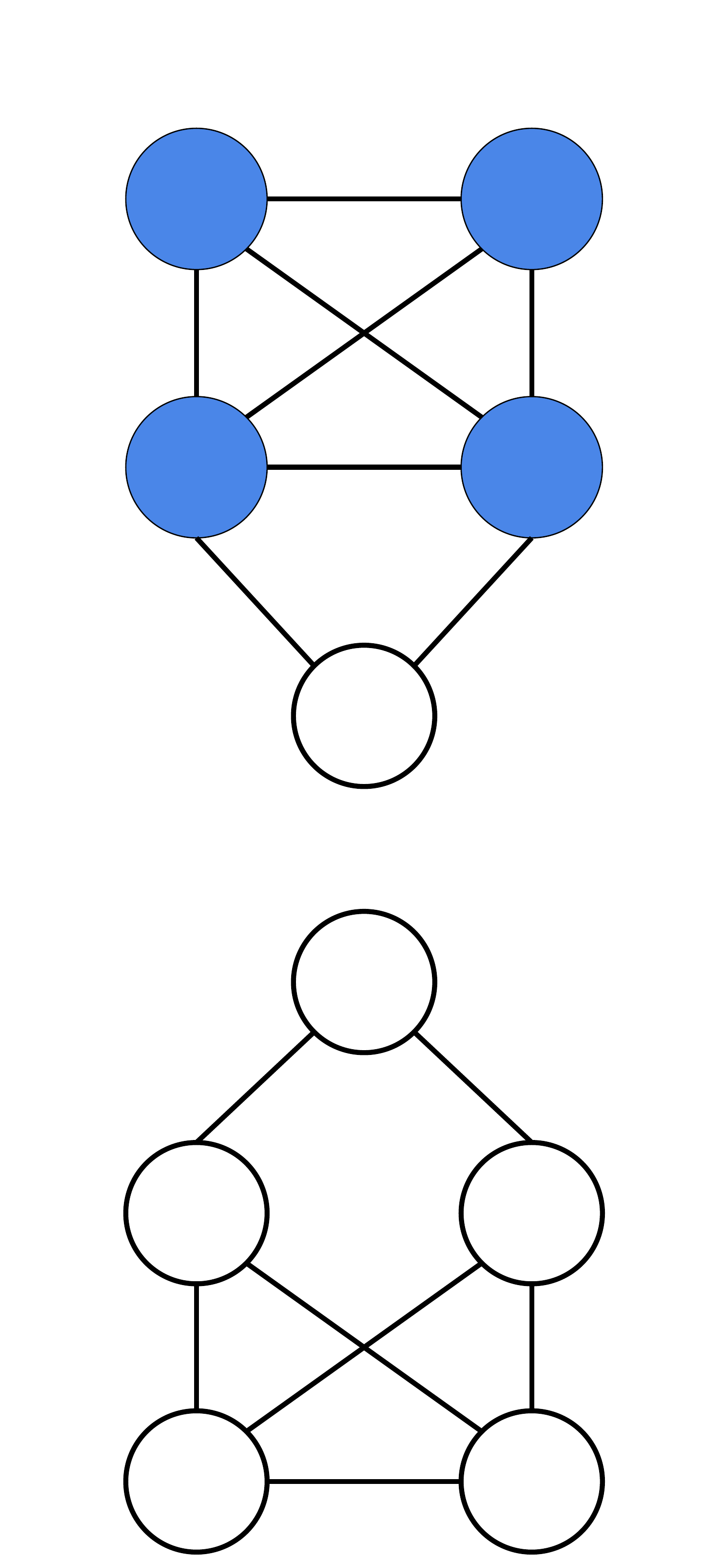}\label{fig:intro_3}}
    \caption{\textbf{ K-core minimization for an illustrative example: (a) Initial graph, where all the vertices are in the $3$-core; (b) Removing $e_1$ causes all the vertices to leave the $3$-core; (c) Removing $e_2$ causes only six vertices to leave the $3$-core. \label{fig:core_intro}}}
    \vspace{-2mm}
\end{figure}



$K$-core minimization (KCM) can be motivated both from the perspective of a centralized agent who protects the structure of a network or an adversary that aims to disrupt it. Moreover, our problem can also be applied to measure network resilience \cite{Laishram2018}.

We illustrate KCM in Figure \ref{fig:core_intro}. An initial graph $G$ (Figure \ref{fig:intro_1}), where all vertices are in the $3$-core, is modified by the removal of a single edge. Graphs $G'$ (Figure \ref{fig:intro_2}) and $G''$ (Figure \ref{fig:intro_2}) are the result of removing $e_1$ and $e_2$, respectively. While the removal of $e_1$ brings all the vertices into a $2$-core,  deleting $e_2$ has a smaller effect---four vertices remain in the 3-core. Our goal is to identify a small set of edges removal of which minimizes the size of the $k$-core.

From a theoretical standpoint, for any objective function of interest, we can define a \textit{search}  (e.g. the $k$-core decomposition) and a corresponding \textit{modification} problem, such as $k$-core minimization. In this paper, we show that, different from its search version \cite{batagelj2011fast}, KCM is NP-hard. Furthermore, there is no polynomial time algorithm that achieves a constant-factor approximation for our problem. Intuitively, the main challenge stems from the strong combinatorial nature of the effects of edge removals. While removing a single edge may have no immediate effect, the deletion of a small number of edges might cause the collapse of the k-core structure. This behavior differs from more popular problems in graph combinatorial optimization, such as submodular optimization, where a simple greedy algorithm provides constant-factor approximation guarantees.

The algorithm for $k$-core minimization proposed in this paper applies the concept of \textit{Shapley values} (SVs), which, in the context of cooperative game theory, measure the contribution of players in  coalitions \cite{shapley1953value}. Our algorithm selects edges with largest Shapley value to account for the joint effect (or cooperation) of multiple edges. Since computing SVs is NP-hard, we approximate them in polynomial time via a randomized algorithm with quality guarantees. 

Recent papers have introduced the KCM problem \cite{zhu2018k} and its vertex version \cite{zhang2017finding}, where the goal is to delete a few vertices such that the $k$-core structure is minimized. However, our work provides a stronger theoretical analysis and more effective algorithms that can be applied to both problems. In particular, we show that our algorithm outperforms the greedy approach proposed in \cite{zhu2018k}.


Our main contributions are summarized as follows:
\begin{itemize}
    \item We study the $k$-core minimization (KCM) problem, which consists of finding a small set of edges, removal of which minimizes the size of the $k$-core structure of a network.
    \item We show that KCM is NP-hard, even to approximate by a constant for $k\geq 3$. We also discuss the parameterized complexity of KCM and show the problem is $W[2]$-hard for the same values of $k$.
    \item Given the above inapproximability result, we propose a randomized Shapley Value based algorithm that efficiently accounts for the interdependence among the candidate edges for removal. 
    \item We show that our algorithm is both accurate and efficient using several datasets. Moreover, we illustrate how KCM can be applied to profile the structural resilience of real networks.
\end{itemize}

\noindent

\section{Problem Definition}
\label{sec:prob_def}

We assume $G(V,E)$ to be an undirected and unweighted graph with sets of vertices $V$ ($|V|=n$) and edges $E$ ($|E|=m$). Let $d(G,u)$ denote the degree of vertex $u$ in $G$. An induced subgraph, $H=(V_H,E_H) \subset G$ is the following: if $u,v \in V_H$ and $(u,v)\in E$ then $(u,v)\in E_H$. The $k$-core \cite{seidman1983network} of a network is defined below.

\begin{definition} \textbf{$k$-Core:} The $k$-core of a graph $G$, denoted by $C_k(G)=(V_k(G),E_k(G))$, is defined as a maximal induced subgraph that has vertices with degree at least $k$.
\end{definition}
Figure \ref{fig:core_definitions} shows an example. The graphs in Figures \ref{fig:kcore_ex11} and \ref{fig:kcore_ex12} are the $2$-core and the $3$-core, respectively, of the initial graph in Figure \ref{fig:kcore_ex1}. Note that, $C_{k+1}(G)$ is a subgraph of $C_k(G)$. Let $\mathcal{E}_G(v)$ denote the core number of the node $v$ in $G$. If $v\in V_k(G)$ and $v\notin V_{k+1}(G)$ then $\mathcal{E}_G(v)=k$. $K$-core decomposition can be performed in time $O(m)$ by recursively removing vertices with degree lower than $k$ \cite{batagelj2011fast}.

 Let $G^B = (V,E\setminus B)$ be the modified graph after deleting a set $B$ with $b$ edges. Deleting an edge reduces the degree of two vertices and possibly their core numbers. The reduction in core number might propagate to other vertices. For instance, the vertices in a simple cycle are in the $2$-core but deleting any edge from the graph moves all the vertices to the $1$-core. Let $N_k(G) =|V_k(G)|$  and $M_k(G) =|E_k(G)|$  be the number of nodes and edges respectively in $C_k(G)$.

\begin{table} [t]
\centering
\small
\vspace{-2mm}
\begin{tabular}{| c | c |}
\hline
\textbf{Symbols} & \textbf{Definitions and Descriptions}\\
\hline
$G(V,E)$ & Given graph (vertex set $V$ and edge set $E$)\\
\hline
$n$ & Number of nodes in the graph\\
\hline
$m$ & Number of edges in the graph\\
\hline
$C_k(G)=(V_k(G),E_k(G))$ & The $k$-core of graph $G$\\
\hline
$N_k(G)$ & $|V_k(G)|$, $\#$nodes in the $k$-core of $G$\\
\hline
$M_k(G)$ & $|E_k(G)|$, $\#$edges in the $k$-core of $G$\\
\hline
$\Gamma$& Candidate set of edges\\
\hline
$b$& Budget \\
\hline
$\mathscr{V}(P)$ & The value of a coalition $P$ \\
\hline
$\Phi_e$ & The Shapley value of an edge $e$ \\
\hline
$P_e(\pi)$  & Set of edges before $e$ in permutation $\pi$ \\
\hline
\end{tabular}
\caption { Frequently used symbols}\label{tab:table_symbol}
\end{table}

\begin{definition} \textbf{Reduced $k$-Core:} A reduced $k$-core, $C_k(G^B)$ is the $k$-core in $G^B$, where $G^B = (V,E\setminus B)$. 
\end{definition}

\begin{example} 
Figures \ref{fig:kcore_ex2} and \ref{fig:kcore_ex3} show an initial graph, $G$ and modified graph $G^B$ (where $B=\{(a,c)\}$) respectively. In $G$, all the nodes are in the $3$-core. Deleting $(a,c)$ brings the vertices $a$ and $c$ to the $2$-core and thus $b$ and $d$ also go to the $2$-core. 
\end{example}

\begin{definition} \textbf{$K$-Core Minimization (KCM):} Given a candidate edge set $\Gamma$, find the set, $B \subset \Gamma$ of $b$ edges to be removed such that $C_k(G^B)$ is minimized, or, $f_k(B)=N_k(G)-N_k(G^B)$ is maximized.
\label{def:kcm}
\end{definition}

\begin{example} 
Figures \ref{fig:kcore_ex2} shows an initial graph, $G$, where all the nodes are in the $3$-core. Deleting $(a,c)$ and $(e,g)$ brings all the vertices to the $2$-core, whereas deleting $(e,c)$ and $(d,f)$ has no effect on the $3$-core structure (assuming $b\!=\!2)$.
\end{example}

Clearly, the importance of the edges varies in affecting the $k$-core upon their removal. Next, we discuss strong inapproximability results for the KCM problem along with its parameterized complexity. 

\begin{figure}[t]
\vspace{-6mm}
    \centering
     \subfloat[Initial graph, $G$]{\includegraphics[width=0.12\textwidth]{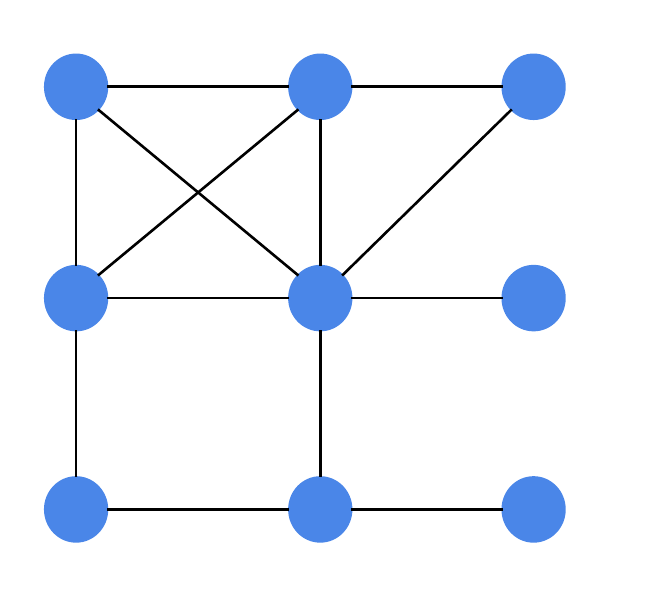}\label{fig:kcore_ex1}}
      \subfloat[The $2$-core of $G$]{\includegraphics[width=0.12\textwidth]{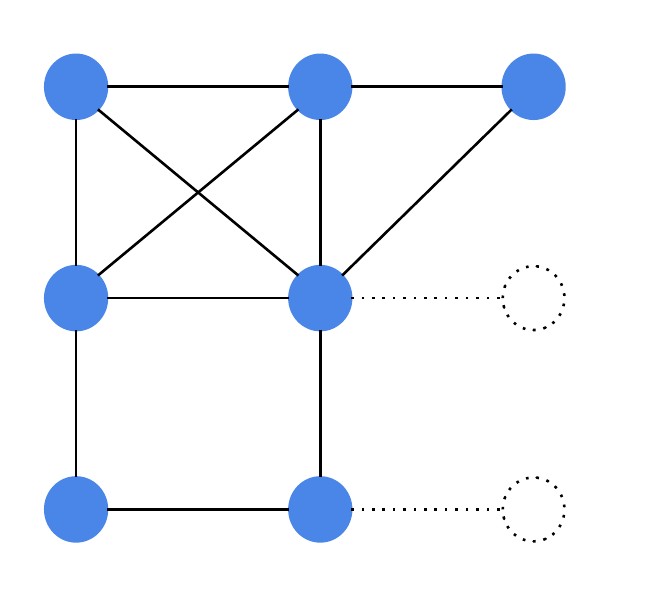}\label{fig:kcore_ex11}}
       \subfloat[The $3$-core of $G$]{\includegraphics[width=0.12\textwidth]{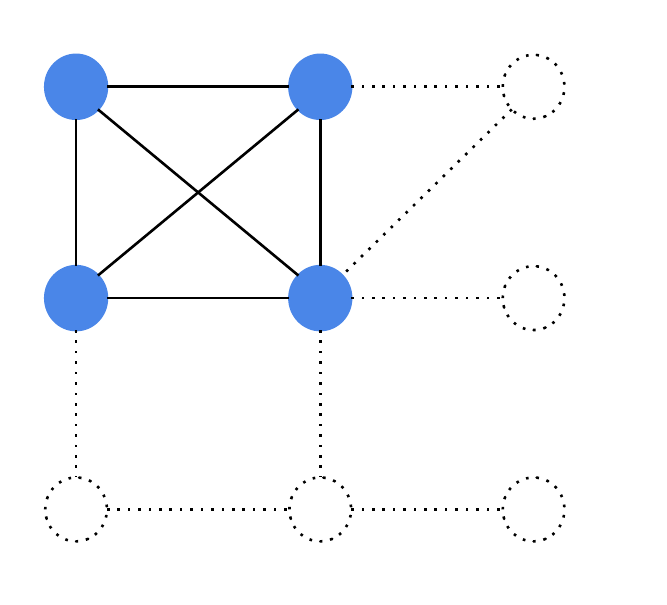}\label{fig:kcore_ex12}}
    \caption{\textbf{ Examples of (a) a graph $G$; (b) its $2$-core; and (c) its $3$-core structures. \label{fig:core_definitions}}}
\end{figure}

\begin{figure}[t]
\vspace{-2mm}
    \centering
      \subfloat[Initial,  $G$]{\includegraphics[width=0.12\textwidth]{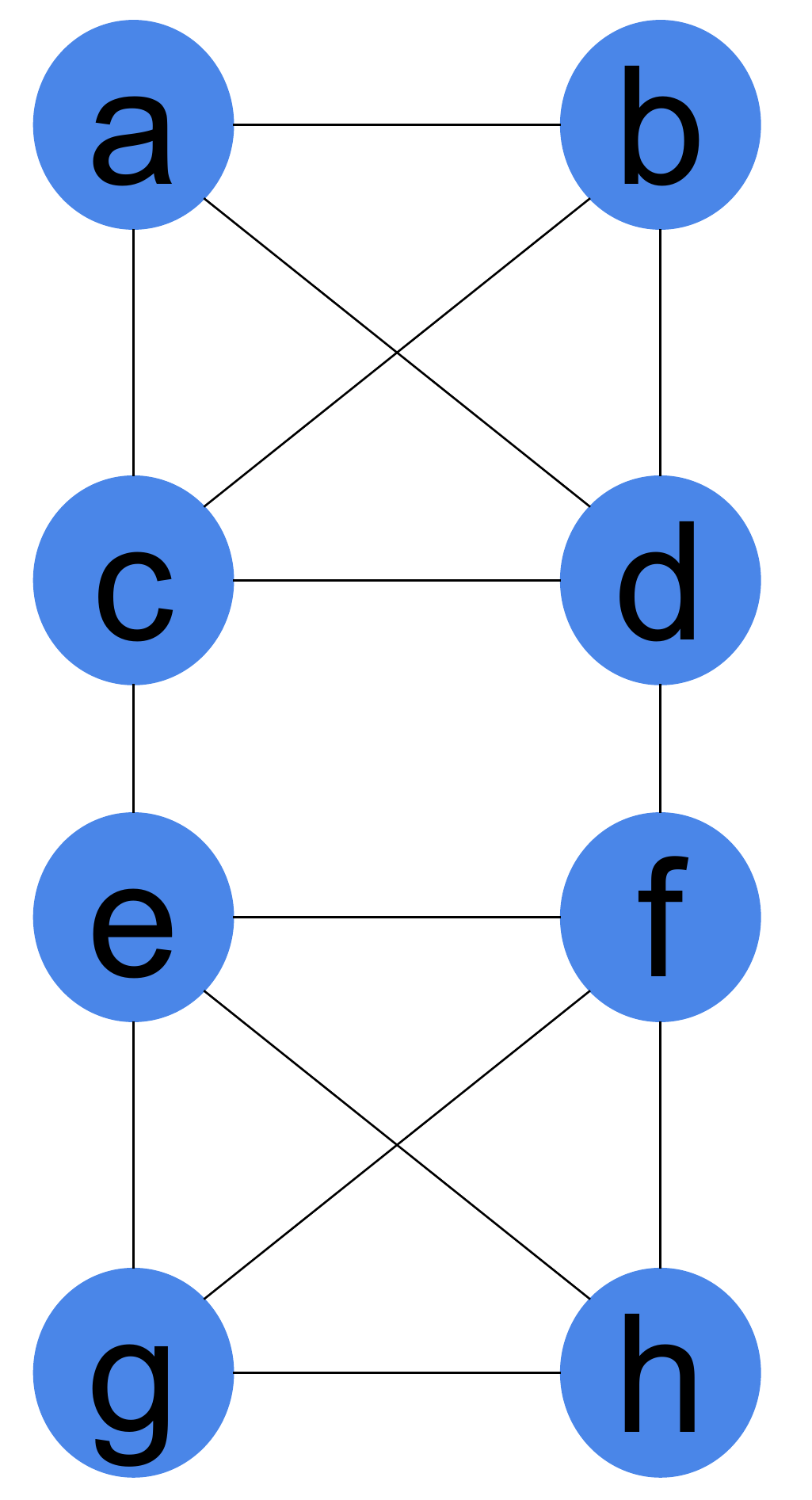}\label{fig:kcore_ex2}}
      \hspace{12mm}
       \subfloat[Modified,  $G^B$]{\includegraphics[width=0.12\textwidth]{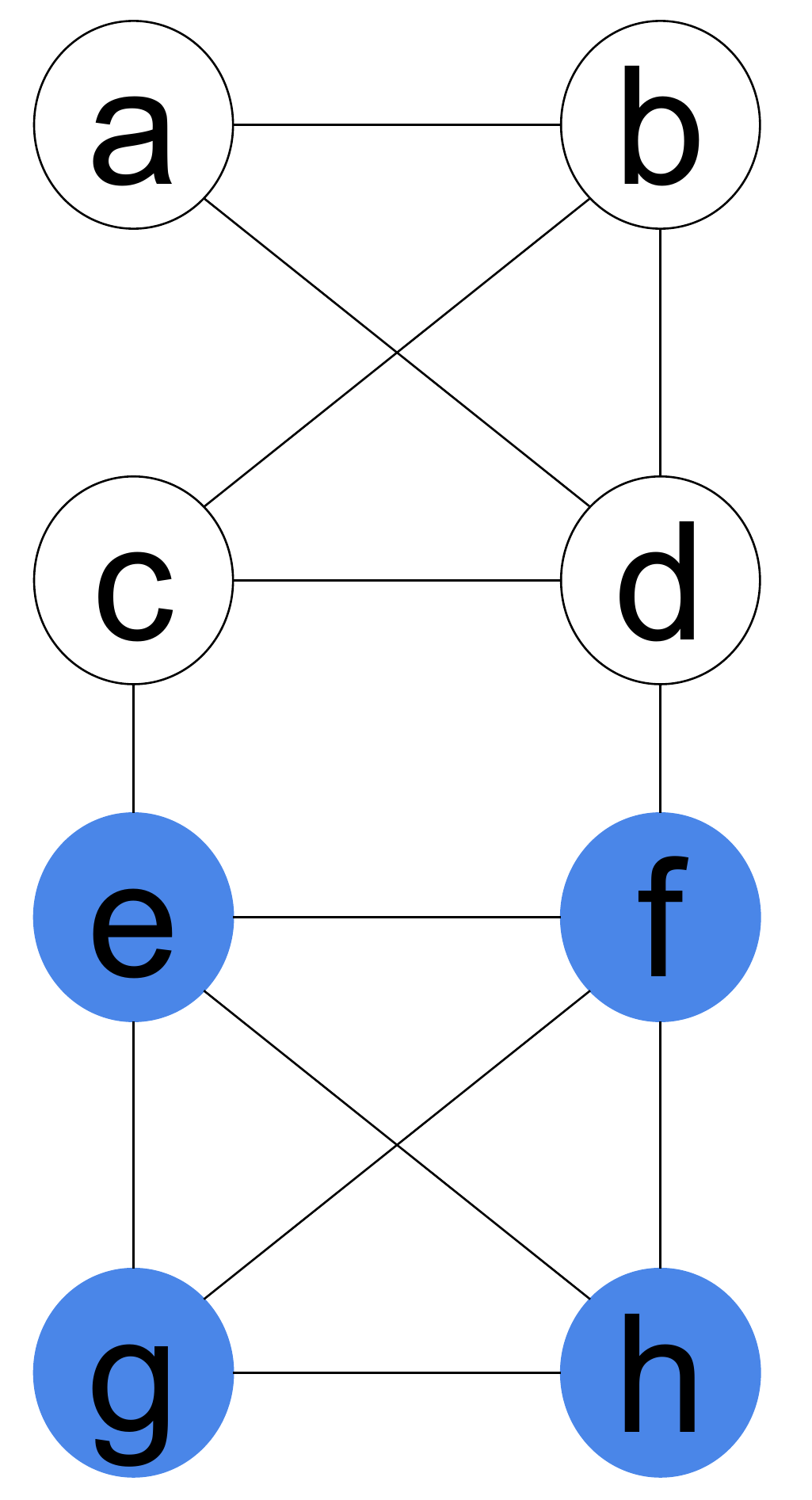}\label{fig:kcore_ex3}}
    \caption{ Example of the changes in the core structure via deletion of an edge: (a) All the nodes are in the $3$-core. (b) In the modified graph, the nodes $\{a,b,c,d\}$ are in the $2$-core. \label{fig:core_examples}}
\end{figure}

\subsection{Hardness and Approximability}
\label{sec:hardness}

The hardness of the KCM problem stems from two major facts: 1) There is a combinatorial number of choices of edges from the candidate set, and 2) there might be strong dependencies in the effects of edge removals (e.g. no effect for a single edge but cascading effects for subsets of edges). We show that KCM is NP-hard to approximate within any constant factor for $k\geq 3$. 

\begin{thm} \label{thm:np_hard_k_1_2}
The KCM problem is NP-hard for $k=1$ and $k=2$.
\end{thm}
\begin{proof}
See the Appendix.
\end{proof}

Notice that a proof for Theorem \ref{thm:np_hard_k_1_2} is also given by \cite{zhu2018k}. However, our proof applies a different construction and was developed independently from (and simultaneously with) this previous work.

\begin{thm} \label{thm:hard_approx}
The KCM problem is NP-hard and it is also NP-hard to approximate within a constant-factor for all $k\geq 3$.
\end{thm}
\begin{proof}
We sketch the proof for $k\!=\!3$ (similar for $k\!>\!3$).

Let $SK(U,\mathcal{S},P,W,q)$ be an instance of the Set Union Knapsack Problem \cite{goldschmidt1994note}, where $U=\{u_1, \ldots u_{n'}\}$ is a set of items, $\mathcal{S}=\{S_1, \ldots S_{m'}\}$ is a set of subsets ($S_i \subseteq U$), $p:\mathcal{S}\to\mathbb{R}_{+}$ is a subset profit function, $w:U\to\mathbb{R}_+$ is an item weight function, and $q \in \mathbb{R}_+$ is the budget. For a subset $\mathcal{A} \subseteq \mathcal{S}$, the weighted union of set $\mathcal{A}$ is $W(\mathcal{A}) = \sum_{e\in \cup_{t\in \mathcal{A}} S_t} w_e$ and $P(\mathcal{A})=\sum_{t \in \mathcal{A}} p_t$. The problem is to find a subset $\mathcal{A}^*\subseteq S$ such that $W(\mathcal{A}^*)\leq q$ and $P(\mathcal{A}^*)$ is maximized.  SK is NP-hard to approximate within a constant factor \cite{arulselvan2014note}. 

 We reduce a version of $SK$ with equal profits and weights (also NP-hard to approximate) to the KCM problem. The graph $G'$ is constructed as follows. For each $u_j \in U$, we create a cycle of $m'$ vertices $Y_{j,1}, Y_{j,2}, \ldots, Y_{j,m'}$ in $V$ and add $(Y_{j,1},Y_{j,2})$, $(Y_{j,2},Y_{j,3}), \ldots, (Y_{j,m'-1}$   $,Y_{j,m'}),(Y_{j,m'},Y_{j,1})$ as edges between them. We also add $5$ vertices $Z_{j,1}$ to $Z_{j,5}$ with eight edges where the four vertices $Z_{j,2}$ to $Z_{j,5}$ form a clique with six edges. The other two edges are $(Z_{j,1},Z_{j,2})$ and $(Z_{j,1},Z_{j,5})$. Moreover, for each subset $S_i$ we create a set of $O((m')^3)$ vertices (sets $X_{i,*}$ are red rectangles in Figure \ref{fig:hardness_ex1}), such that each node has exactly degree $3$, and add one more node $X_{i,1}$ with two edges incident to two vertices in $X_{i,*}$ from $X_{i,1}$. In the edge set $E$, an edge $(X_{i,1},Y_{j,i})$ will be added if $u_j\in S_i$. Additionally, if $u_j\notin S_i$, the edge $(Y_{j,i},Z_{j,1})$ will be added to $E$. Figure \ref{fig:hardness_ex1} illustrates our construction for a set $S_1=\{u_1,u_2\},S_2=\{u_1,u_3\},S_3=\{u_2\}$.
 
In KCM, the number of edges to be removed is the budget, $b$. The candidate set of edges, $\Gamma$ is the set of all the edges with form  $(Y_{j,1},Y_{j,2})$. Initially all the nodes in $G'$ are in the $3$-core. Our claim is, for any solution $\mathcal{A}$ of an instance of $SK$ there is a corresponding solution set of edges, $B$ (where $|B|=b$) in $G'$ of the KCM problem, such that $f_3(B)=P(\mathcal{A})+b(m'+1)$ if the edges in $\mathcal{A}$ are removed. 

 The $m'$ nodes in any $Y_j$ and the node $Z_{j,1}$ will be in the $2$-core if the edge $(Y_{j,1},Y_{j,2)}$ gets removed. So, removal of any $b$ edges from $\Gamma$ enforces $b(m'+1)$ nodes to go to the $2$-core. But the node $X_{i,1}$ and each node in $X_{i,*}$ ($O((m')^3)$ nodes) will be in the $2$-core iff all its neighbours in $Y_{j,i}$s go to the $2$-core after the removal of $b$ edges in $\Gamma$. Thus, an optimal solution $B^*$ will be $f_3(B^*)=P(\mathcal{A^*})+b(m'+1)$ where $\mathcal{A^*}$ is the optimal solution for SUKP. For any non-optimal solution $B$,  $f_3(B)=P(\mathcal{A})+b(m'+1)$ where $\mathcal{A}$ is also non-optimal solution for SUKP.  As $P(\mathcal{A^*})$ is at least $O((m')^3)$ by construction (i.e. $P(\mathcal{A^*})\gg b(m'+1)$), and $\frac{P(\mathcal{A^*})}{P(\mathcal{A})}$ cannot be within a constant factor, $\frac{f_3(B^*)}{f_3(B)}$ will also not be within any constant factor.
\end{proof}

Theorem \ref{thm:hard_approx} shows that there is no polynomial-time constant-factor approximation for KCM when $k\!\geq\!3$. This contrasts with well-known NP-hard graph combinatorial problems in the literature \cite{kempe2003maximizing}. In the next section, we explore the hardness of our problem further in terms of exact
exponential algorithms with respect to the parameters.

\begin{figure}[t]
\vspace{-1mm}
    \centering
    {\includegraphics[width=0.45\textwidth]{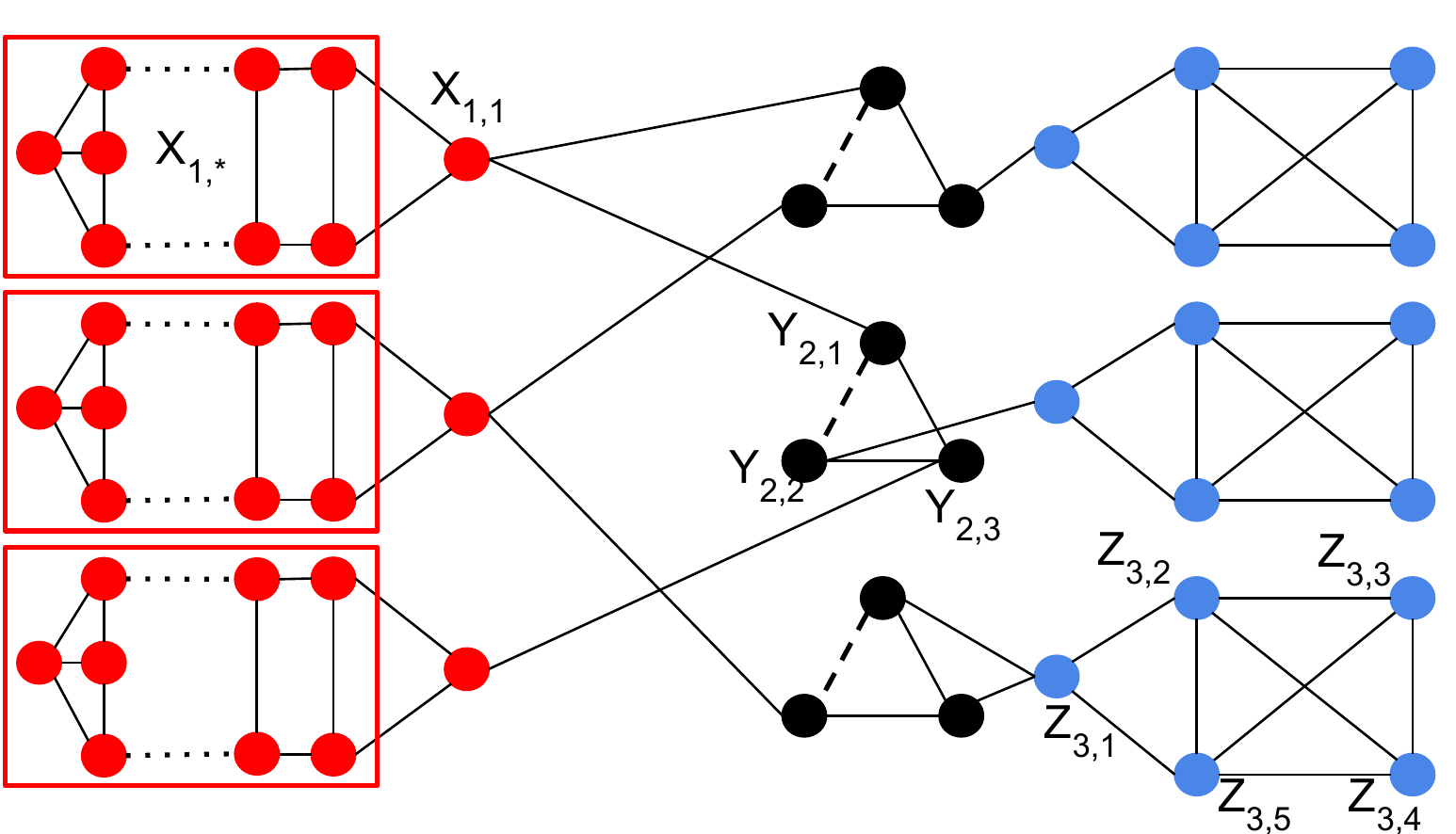}}
    \caption{\textbf{Example construction for hardness reduction from SK where  $U=\{u_1,u_2,u_3\}, S=\{S_1,S_2,S_3\}, S_1=\{u_1,u_2\},S_2=\{u_1,u_3\},S_3=\{u_2\}$. \label{fig:hardness_ex1}}}
\end{figure}

\subsection{Parameterized Complexity}
There are several NP-hard problems with exact solutions via algorithms that run in exponential time in the size of the parameter. For instance, the NP-hard Vertex Cover can be solved via an exhaustive search algorithm in time $2^{b_1}{n_1}^{O(1)}$ \cite{balasubramanian1998improved}, where $b_1$ and $n_1$ are budget and the size of the graph instance respectively. Vertex cover is therefore fixed-parameter tractable (FPT) \cite{flum2006parameterized}, and if we are only interested in small $b_1$, we can solve the problem in polynomial time. We investigate whether the KCM problem is also in the FPT class.

A parameterized problem instance is comprised of an instance $X$ in the
usual sense, and a parameter $b$. A problem with parameter $b$ is called fixed parameter tractable (FPT) \cite{flum2006parameterized}
if it is solvable in time $g(b) \times p(|X|)$, where $g$ is an arbitrary function of $b$ and $p$ is a polynomial in the
input size $|X|$. Just as in NP-hardness,
 there exists a hierarchy of
complexity classes above FPT. Being hard for one of these classes is an evidence that the problem is unlikely to be FPT. Indeed, assuming the
Exponential Time Hypothesis, a problem which is $W[1]$-hard does not belong to FPT.
The main classes in this hierarchy are: FPT$\subseteq W[1] \subseteq W[2]\subseteq\ldots W[P]\subseteq XP$. Generally speaking, the problem is harder when it belongs to a higher $W[.]$-hard class in terms of the parameterized complexity. For instance, \textit{dominating set} is in $W[2]$ and is considered to be harder than  \textit{maximum independent set}, which is in $W[1]$.

\begin{definition}
\textbf{Parameterized Reduction \cite{flum2006parameterized}: }
Let $P_1$ and $P_2$ be parameterized problems. A parameterized reduction from $P_1$ to $P_2$ is an algorithm that, given an instance $(X_1, b_1)$ of $P_1$, outputs an instance $(X_2, b_2)$ of $P_2$ such that: (1) $(X_1, b_1)$ is a yes-instance of $P_1$ iff $(X_2, b_2)$ is a yes-instance of $P_2$; (2) $b_2\leq h(b_1)$ for some computable (possibly exponential) function $h$; and (3) the running time of the algorithm is $g(b_1)\cdot|X|^{O(1)}$ for a computable function $g$.

\end{definition}

\begin{thm} \label{thm: param_approx}
The KCM problem is not in FPT, in fact, it is in $W[2]$ parameterized by $b$ for $k\geq 3$.
\end{thm}
\begin{proof}
We show a parameterized reduction from the Set Cover problem. The Set Cover problem is known to be $W[2]$-hard \cite{bonnet2016parameterized}. The details on the proof are given in the Appendix.
\end{proof}

\begin{thm}\label{cor:tmcv_para_d}
 The KCM problem is para-NP-hard parameterized by $k$.
\end{thm}

This can be proven from the fact that our problem KCM is NP-hard even for constant $k$. Motivated by these strong hardness and inapproximability results, we next consider some practical heuristics for the KCM problem. 

\section{Algorithms}
\label{sec:algo}

According to Theorems \ref{thm:hard_approx} and \ref{thm: param_approx}, an optimal solution--- or  constant-factor approximation---for $k$-core minimization requires enumerating all possible size-$b$ subsets from the candidate edge set, assuming $P\!\neq\!NP$. In this section, we propose efficient heuristics for KCM. 


\subsection{Baseline: Greedy Cut~\cite{zhu2018k}}
\label{sec:algo_greedy}

For KCM, only the current $k$-core of the graph, $\mathscr{G}(V_k,E_k)=C_k(G)$ ($|V_k|=N_k$,$|E_k|=M_k$), has to be taken into account. Remaining nodes will already be in a lower-than-$k$-core and can be removed. We define a vulnerable set $VS_k(e,\mathscr{G})$ as those nodes that would be demoted to a lower-than-$k$-core if edge $e$ is deleted from the current core graph $\mathscr{G}$. Algorithm \ref{alg:GC} (GC) is a greedy approach for selecting an edge set $B$ ($|B|=b$) that maximizes the $k$-core reduction, $f_k(B)$. In each step, it chooses the edge that maximizes $|VS_k(e,\mathscr{G})|$  (step $3$-$4$) among the candidate edges $\Gamma$. The specific procedure for computing $VS_k(e,\mathscr{G})$ (step $3$), $LocalUpdate$  and their running times ($O(M_k+N_k)$) are described in the Appendix. The overall running time of GC is $O(b|\Gamma| (M_k+N_k))$. 



\begin{algorithm}[t]
\caption{Greedy Cut (GC)}
\label{alg:GC}
\KwInput{$G,k,b$}
\KwOutput{$B$: Set of edges to delete}
$B\leftarrow \emptyset,max \leftarrow -\infty, \mathscr{G}\leftarrow C_k(G)$\\
\While{$|B|<b$}{
$e^*\leftarrow \argmax_{e \in \mathscr{G}(E_k)\setminus B} |computeVS(e=(u , v),\mathscr{G},k)|$\\
$B\leftarrow B \cup \{e^*\}$ \\
LocalUpdate$(e,\mathscr{G},k)$ \\
}
\textbf{return} $B$
\end{algorithm}

\subsection{Shapley Value Based Algorithm}
\label{sec:shapley_algo}
The greedy algorithm discussed in the last section is unaware of some dependencies between the candidates in the solution set. For instance, in Figure \ref{fig:kcore_ex2}, all the
  edges have same importance (the value is $0$) to destroy the $2$-core structure. In this scenario, GC will choose an edge arbitrarily. However, removing an optimal set of seven edges can make the graph a tree ($1$-core). To capture these dependencies, we adopt a cooperative game theoretic concept named Shapley Value \cite{shapley1953value}. Our goal is to make a coalition of edges (players) and divide the total gain by this coalition equally among the edges inside it.

\subsubsection{Shapley Value}

The Shapley value of an edge $e$ in the context of KCM is defined as follows. Let the value of a coalition $P$ be $\mathscr{V}(P)= f_k(P)=N_k(G)-N_k(G^P)$. Given an edge $e\in \Gamma$ and a subset $P\subseteq \Gamma$ such that $e \notin P$, the marginal contribution of $e$ to $P$ is:
\begin{equation}
\mathscr{V}(P\cup \{e\}) - \mathscr{V}(P),\quad \forall P \subseteq \Gamma
\end{equation}
Let $\Omega$ be the set of all $|\Gamma|!$ permutations of all the edges in $\Gamma$ and $P_e(\pi)$ be the set of all the edges that appear before $e$ in a permutation $\pi$. The Shapley value of $e$ the average of its marginal contributions to the edge set that appears before $e$ in all the permutations:
\begin{equation} \label{eq:SV_all}
   \Phi_e=\frac{1}{|\Gamma|!} \sum_{\pi \in \Omega} \mathscr{V}(P_e (\pi)\cup \{e\}) - \mathscr{V}(P_e (\pi))
\end{equation}

Shapley values capture the importance of an edge inside a set (or coalition) of edges. However, computing Shapley value requires considering $O(|\Gamma|!)$ permutations. Next we show how to efficiently approximate the Shapley value for each edge via sampling.

\begin{algorithm}[h] 
\caption{Shapley Value Based Cut (SV)}
\label{alg:SV}
\KwInput{$G,k,b$}
\KwOutput{$B$: Set of edges to delete}
 Initialize all $\Phi'_e$ as $0$, $\forall e \in \Gamma$ \\
Generate $S=O(\frac{\log{\Gamma}}{\epsilon^2})$ random permutations of edges\\
 $B\leftarrow \emptyset, \mathscr{G}\leftarrow C_k(G)$\\
\For{$\pi \in S$}{
    \For{$e =(u,v) \in \Gamma$} {
      $\Phi'_e \leftarrow \Phi'_e+(\mathscr{V}(P_e (\pi)\cup \{e\}) - \mathscr{V}(P_e (\pi)))$ 
        
    }
}
 $\Phi'_e\leftarrow \frac{\Phi'_e}{|S|}$, $\forall e \in \Gamma$ \\
 Select top $b$ $\Phi'_e$ edges from $B$\\
\textbf{return} $B$
\end{algorithm}

\subsubsection{Approximate Shapley Value Based Algorithm} \label{sec:approx_SV}
Algorithm \ref{alg:SV} (Shapley Value Based Cut, SV) selects the best $b$ edges according to their approximate Shapley values based on a sampled set of permutations, $S$. 
For each permutation in $S$, we compute the marginal gains of all the edges. These marginal gains are normalized by the sample size, $s$. In terms of time complexity, steps 4-6 are the dominating steps and take $O(s|\Gamma|(N_k+M_k))$ time, where $N_k$ and $M_k$ are the number of nodes and edges in $C_k(G)$, respectively. Note that similar sampling based methods have been introduced for different applications \cite{castro2009polynomial, maleki2013bounding} (details are in Section \ref{sec:prev_work}).

\subsubsection{Analysis} \label{sec:approx_SV_analysis}
In the previous section, we presented a fast sampling algorithm (SV) for $k$-core minimization using Shapley values. Here, we study the quality of the approximation provided by SV as a function of the number of samples. We show that our algorithm is nearly optimal with respect to each Shapley value with high probability. More specifically, given $\epsilon >0$ and $\delta < 1$, SV takes $p(\frac{1}{\epsilon},\frac{1}{\delta})$ samples, where $p$ is a polynomial in $\frac{1}{\epsilon},\frac{1}{\delta}$, to approximate the Shapley values within $\epsilon$ error with probability $1-\delta$. 

 We sample uniformly with replacement, a set of permutations $S$ ($|S|=s$) from the set of all permutations, $\Omega$. Each permutation is chosen with probability $\frac{1}{|\Omega|}$. Let $\Phi'_e$ be the approximate Shapley value of $e$ based on $S$. $X_i$ is a random variable that denotes the marginal gain in the $i$-th sampled permutation. So, the estimated Shapley value is $\Phi'_e = \frac{1}{s} \sum_{i=1}^s X_i$. Note that $\mathbb{E}[\Phi'_e]= \Phi_e$.

\begin{theorem} \label{thm:approx_SV}
Given $\epsilon$ $(0<\epsilon<1)$, a positive integer $\ell$, and a sample of independent permutations $S, |S|=s$, where 
$s \geq \frac{(\ell+1)\log{|\Gamma|}}{2\epsilon^2}$; then $\forall e \in \Gamma$:
\begin{equation*}
Pr (|\Phi'_e- \Phi_e | < \epsilon \cdot N_k) \geq 1- 2|\Gamma|^{-\ell}
\end{equation*}
 where $N_k$ denotes the number of nodes in $C_k (G)$.
\end{theorem}

\begin{proof}
We start by analyzing the Shapley value of one edge. Because the samples provide an unbiased estimate and are i.i.d., we can apply \emph{Hoeffding's inequality}~\cite{hoeff1963} to bound the error for edge $e$:
\begin{equation}
Pr[|\Phi'_e-\Phi_e| \geq \epsilon \cdot Q_e ]\leq \delta 
\end{equation}
where $\delta=2\exp\left(-\frac{2s^2\epsilon^2 Q^2_e}{\mathcal{R}}\right)$, $\mathcal{R} =\sum\limits_{i=1}^{s}(b_i-a_i)^2$, and each $X_i$ is strictly bounded by the intervals $[a_i, b_i]$. Let $Q_e= Max\{\mathscr{V}(P_e (\pi)\cup \{e\}) - \mathscr{V}(P_e (\pi))| \pi \in \Omega\}$ be the maximum gain for $e$ in any permutation. Then, $\mathcal{R}< sQ^2_e$, as for any $X_i$ the minimum and maximum values are $0$ and $Q_e$ respectively. As a consequence:
\begin{equation*}
\delta =2\exp\left(-\frac{2s^2\epsilon^2 Q^2_e}{\mathcal{R}}\right) < 2\exp\left(-\frac{2s^2\epsilon^2 Q^2_e}{sQ^2_e}\right)  =  2\exp\left(-2s\epsilon^2 \right)
\end{equation*}

Thus, the following holds for each edge $e$:
 \begin{equation*}
 Pr[|\Phi'_e-\Phi_e| \geq \epsilon \cdot Q_e ] < 2\exp\left(-2s\epsilon^2 \right)
 \end{equation*}
 
 Using the above equation we compute a joint sample bound for all edges $e\in \Gamma$. Let $\Gamma=\{e_1,e_2,...,e_{|\Gamma|}\}$ and $E_i$ be the event that $|\Phi'_{e_i}-\Phi_{e_i}| \geq \epsilon \cdot Q_{e_i}$. So, 
 $Pr[E_i]= Pr[|\Phi'_{e_i}-\Phi_{e_i}| \geq \epsilon \cdot Q_{e_i} ] < 2\exp\left(-2s\epsilon^2 \right) $. Similarly, one can prove that 
 $Pr[|\Phi'_{e_i}-\Phi_{e_i}| \geq \epsilon \cdot N_k ] \leq  \delta'$,
 where $\delta' =2\exp\left(-\frac{2s^2\epsilon^2 N^2_k}{\mathcal{R}}\right) <  2\exp\left(-2s\epsilon^2 \right)$, as  $\mathcal{R}< sN^2_k$.
 
 Applying union bound ($Pr(\cup_{i}E_i)\leq \sum_{i}Pr(E_i)$), for all edges in $ \Gamma$, i.e., $\forall i \in \{1,2,...|\Gamma|\}$, we get that:
 \begin{equation*}
 Pr[|\Phi'_{e_i}-\Phi_{e_i}| \geq \epsilon \cdot N_k ] < 2|\Gamma|\exp\left(-2s\epsilon^2 \right)
 \end{equation*}
 
 By choosing $s \geq \frac{(\ell+1)\log{|\Gamma|}}{2\epsilon^2} $, $\forall i \in \{1,2,...|\Gamma|\}$,
 \begin{equation*}
 Pr[|\Phi'_{e_i}-\Phi_{e_i}| \geq \epsilon \cdot N_k ] < \frac{2}{|\Gamma|^\ell}, \quad \text{or,}
 \end{equation*}
  \begin{equation*}
 Pr[|\Phi'_{e_i}-\Phi_{e_i}| < \epsilon \cdot N_k) \geq 1- 2|\Gamma|^{-\ell}
 \end{equation*}
  This ends the proof.
\end{proof}

Next, we apply Theorem \ref{thm:approx_SV} to analyze the quality of a set $B$ produced by Algorithm \ref{alg:SV} (SV), compared with the result of an exact algorithm (without sampling). Let the exact Shapley values of top $b$ edges be $\Phi^o_{B}=\{\Phi_{O1},\Phi_{O2},\Phi_{O3},...,\Phi_{Ob}\}$ where $\Phi_{O1}\geq \Phi_{O2}\geq...\geq\Phi_{Ob}$. The set produced by Algorithm \ref{alg:SV} (SV) has Shapley values, $\Phi^a_{B}=\{\Phi_{A1},\Phi_{A2},\Phi_{A3},...,\Phi_{Ab}\}$ where $\Phi_{A1}\geq \Phi_{A2}\geq...\geq\Phi_{Ab}$. We can prove the following result regarding the SV algorithm.

\begin{cor}\label{cor:anyset}
For any  $i, \Phi_{Oi} \in \Phi^o_{B}$ and $\Phi_{Ai} \in \Phi^a_{B}$, $\epsilon$ $(0<\epsilon<1)$, positive integer $\ell$, and a sample of independent permutations $S, |S|=s$, where 
$s \geq \frac{(\ell+1)\log{|\Gamma|}}{2\epsilon^2}$:
\begin{equation*}
Pr (|\Phi_{Oi}- \Phi_{Ai} | < 2\epsilon \cdot N_k) \geq 1- 2|\Gamma|^{-\ell}
\end{equation*}
 where $N_k$ denotes the number of nodes in $C_k (G)$.

\end{cor}

\begin{proof} For all edges $e \in \Gamma$, Theorem \ref{thm:approx_SV} shows that $Pr (|\Phi'_e- \Phi_e | < \epsilon \cdot N_k) \geq 1- 2|\Gamma|^{-\ell}$. So, with probability $1- 2|\Gamma|^{-\ell}$, $|\Phi'_{Oi}- \Phi_{Oi} | < \epsilon \cdot N_k$ and $ |\Phi'_{Ai}- \Phi_{Ai} | < \epsilon \cdot N_k$. As $\Phi'_{Ai} > \Phi'_{Oi} $, $ |\Phi_{Oi}- \Phi_{Ai} | < 2\epsilon \cdot N_k$  with the same probability. 
\end{proof}

At this point, it is relevant to revisit the hardness of approximation result from Theorem \ref{thm:hard_approx} in the light of Corollary \ref{cor:anyset}. First, SV does not directly minimize the KCM objective function (see Definition \ref{def:kcm}). Instead, it provides a score for each candidate edge $e$ based on how different permutations of edges including $e$ minimize the KCM objective under the assumption that such scores are divided fairly among the involved edges. Notice that such assumption is not part of the KCM problem, and thus Shapley values play the role of a heuristic. Corollary \ref{cor:anyset}, which is a polynomial-time randomized approximation scheme (PRAS) type of guarantee instead of a constant-factor approximation, refers to the exact Shapley value of the top $b$ edges, and not the KCM objective function. We evaluate how SV performs regarding the KCM objective in our experiments. 

\subsubsection{Generalizations} Sampling-based approximate Shapley values can also be applied to other relevant combinatorial problems on graphs for which the objective function is not submodular.  Examples of these problems include $k$-core anchoring  \cite{bhawalkar2015preventing}, influence minimization \cite{kimura2008minimizing}, and network design \cite{dilkina2011}).

\begin{table}[t]
\centering
\begin{tabular}{| c | c | c |c|c|}
\hline
\textbf{Dataset Name}& \textbf{$|V|$} & \textbf{$|E|$} & \textbf{$k_{max}$}  & Type\\
\hline
\textbf{Yeast}& 1K & 2.6K & 6 & Biological \\
\hline
\textbf{Human}& 3.6K & 8.9K & 8 & Biological  \\
\hline
\textbf{email-Enron (EE)}& 36K & 183K & 42 & Email \\
\hline
\textbf{Facebook (FB)}& 60K & 1.5M & 52 & OSN \\
\hline
\textbf{web-Stanford (WS)}& 280K & 2.3M& 70 & Webgaph\\
\hline
\textbf{DBLP (DB)}& 317K  & 1M & 113 & Co-authorship\\
\hline
\textbf{com-Amazon (CA)}& 335K & 926K & 6 & Co-purchasing\\
\hline
\textbf{Erdos-Renyi (ER)}& 60K & 800K & 19 & Synthetic\\
\hline
\end{tabular}

\caption{Dataset descriptions and statistics. The value of $k_{max}$ (or degeneracy) is the largest $k$ among all the values of $k$ for which there is a $k$-core in the graph. \label{table:data_description}}
\vspace{-3mm}
 \end{table}
 
\section{Experiments}
\label{sec:exp}

In this section, we evaluate the proposed Shapley Value Based Cut (SV) algorithm for k-core minimization against baseline solutions. Sections \ref{sec::effect_sv} and \ref{sec:running_time} are focused on the quality results (k-core minimization) and the running time of the algorithms, respectively. Moreover, in Section \ref{sec:others}, we show how k-core minimization can be applied in the analysis of the structural resilience of networks.

\begin{figure*}[ht]
\vspace{-3mm}
    \centering
     \subfloat[DB]{\includegraphics[width=0.24\textwidth]{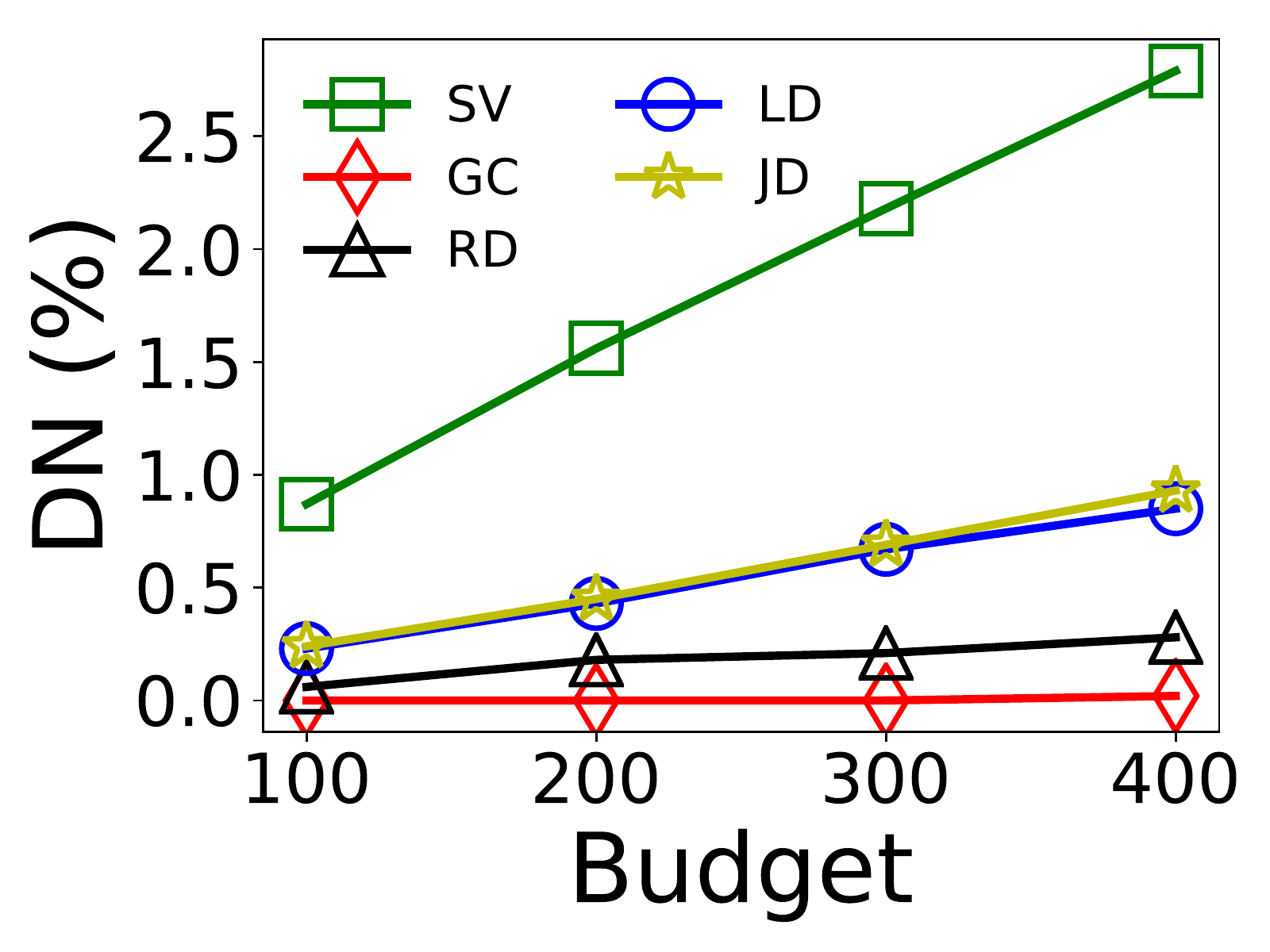}\label{fig:sv_db_budget}}
    \hspace{1mm}
    \subfloat[WS]{\includegraphics[width=0.24\textwidth]{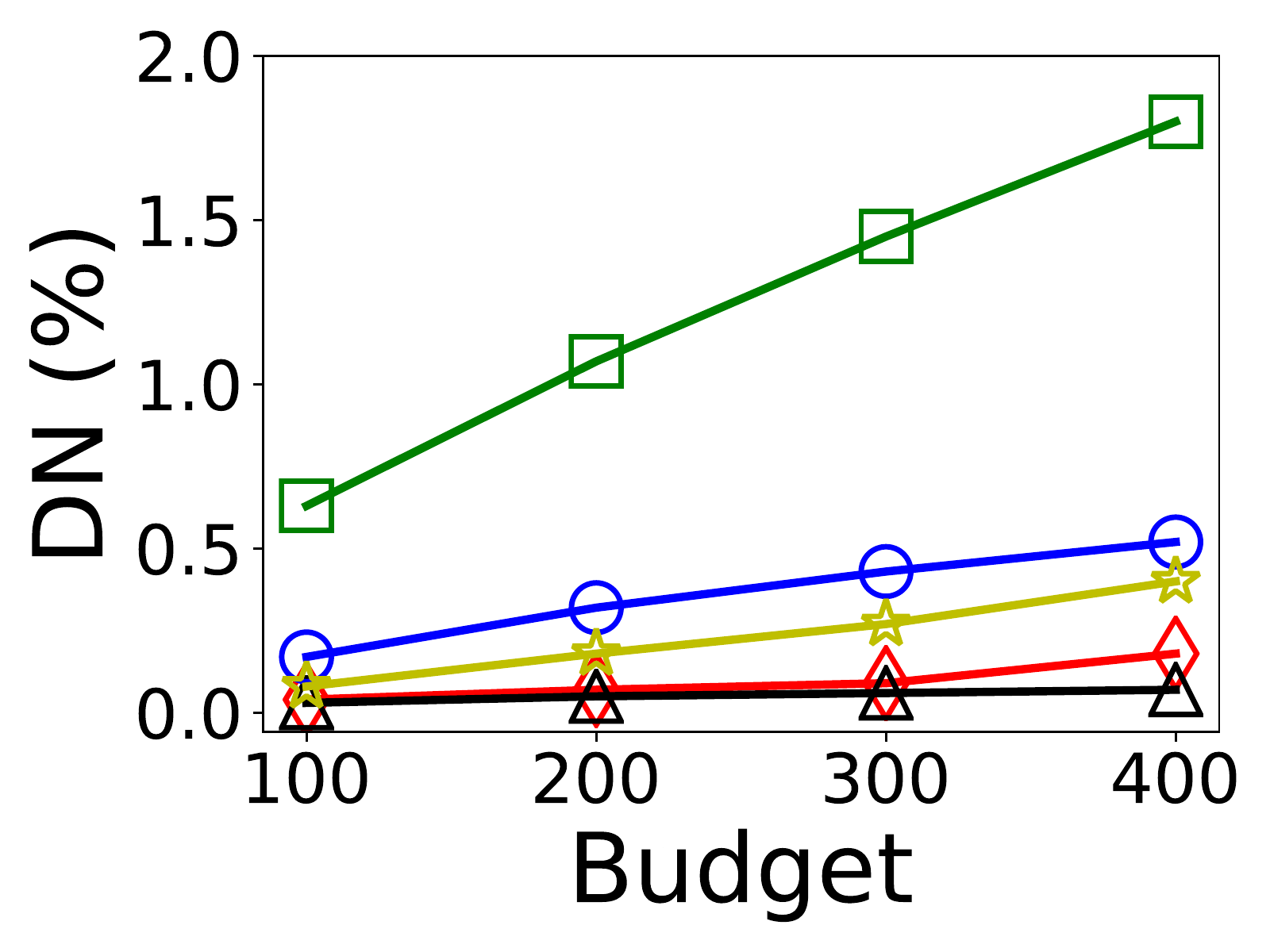}\label{fig:sv_ws_budget}}
    \hspace{1mm}
    \subfloat[EE]{\includegraphics[width=0.24\textwidth]{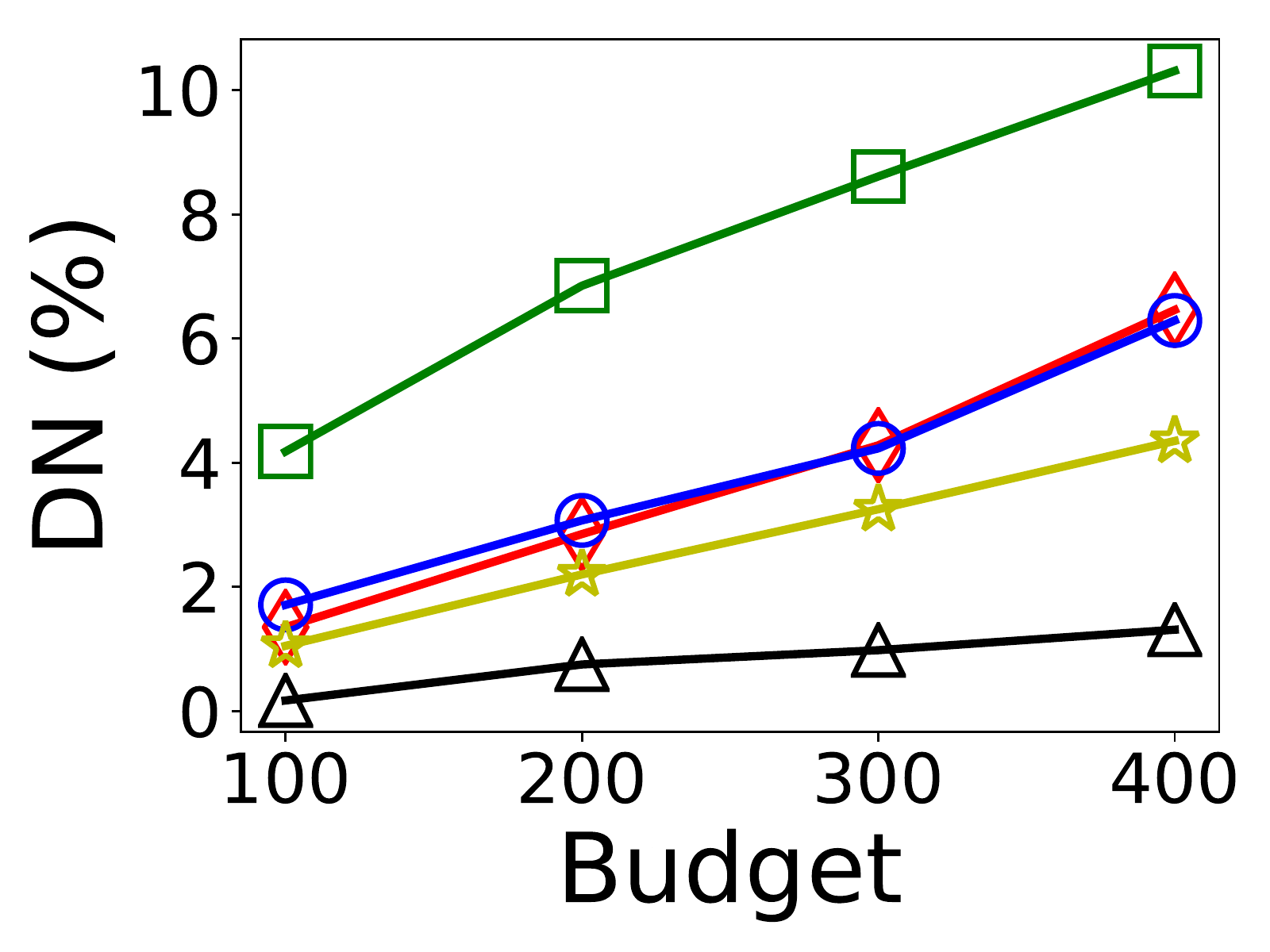}\label{fig:sv_ee_budget}}
    \hspace{1mm}
    \subfloat[FB]{\includegraphics[width=0.24\textwidth]{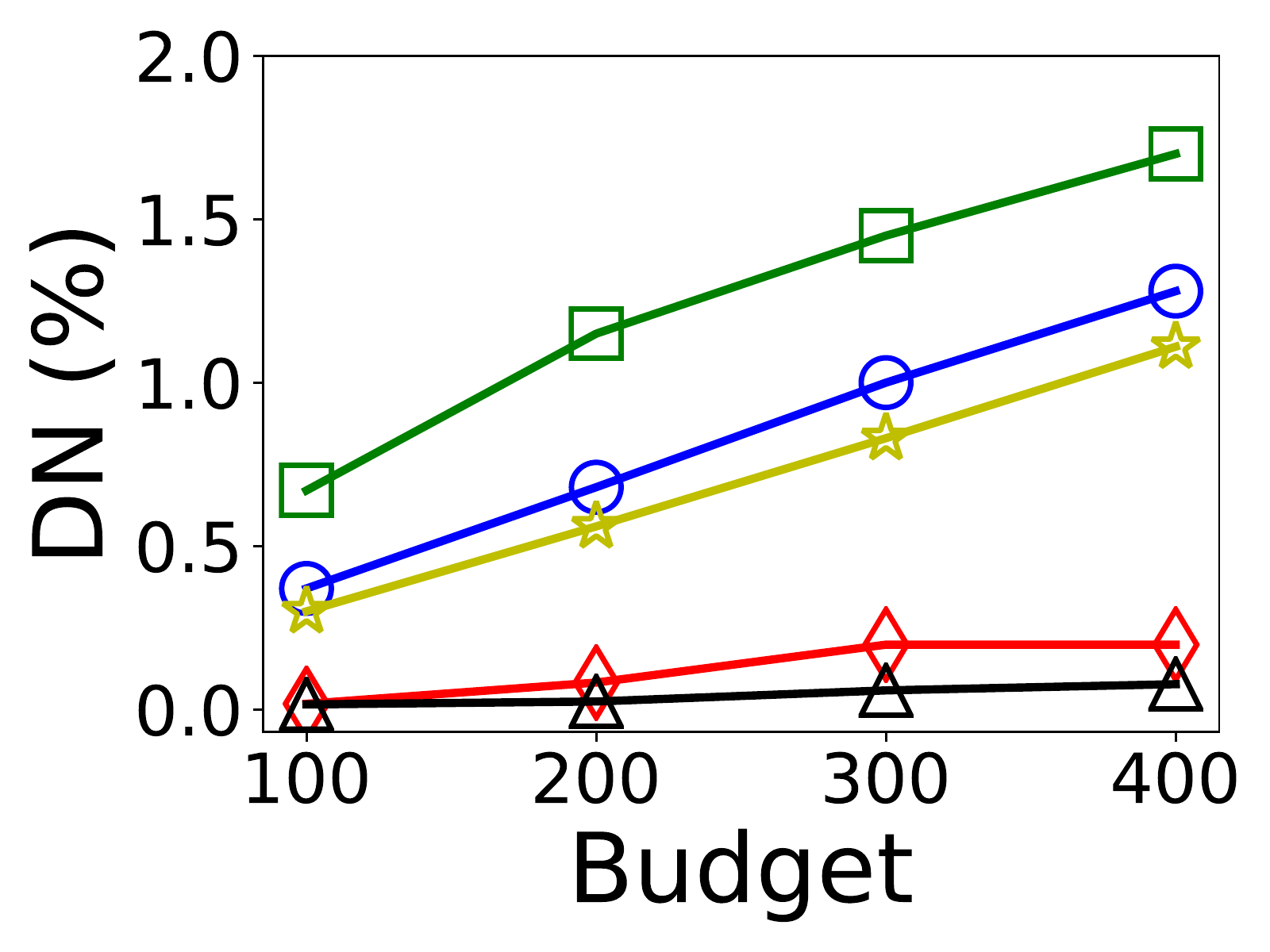}\label{fig:sv_fb_budget}}
    \hspace{1mm}
    \subfloat[FB]{\includegraphics[width=0.24\textwidth]{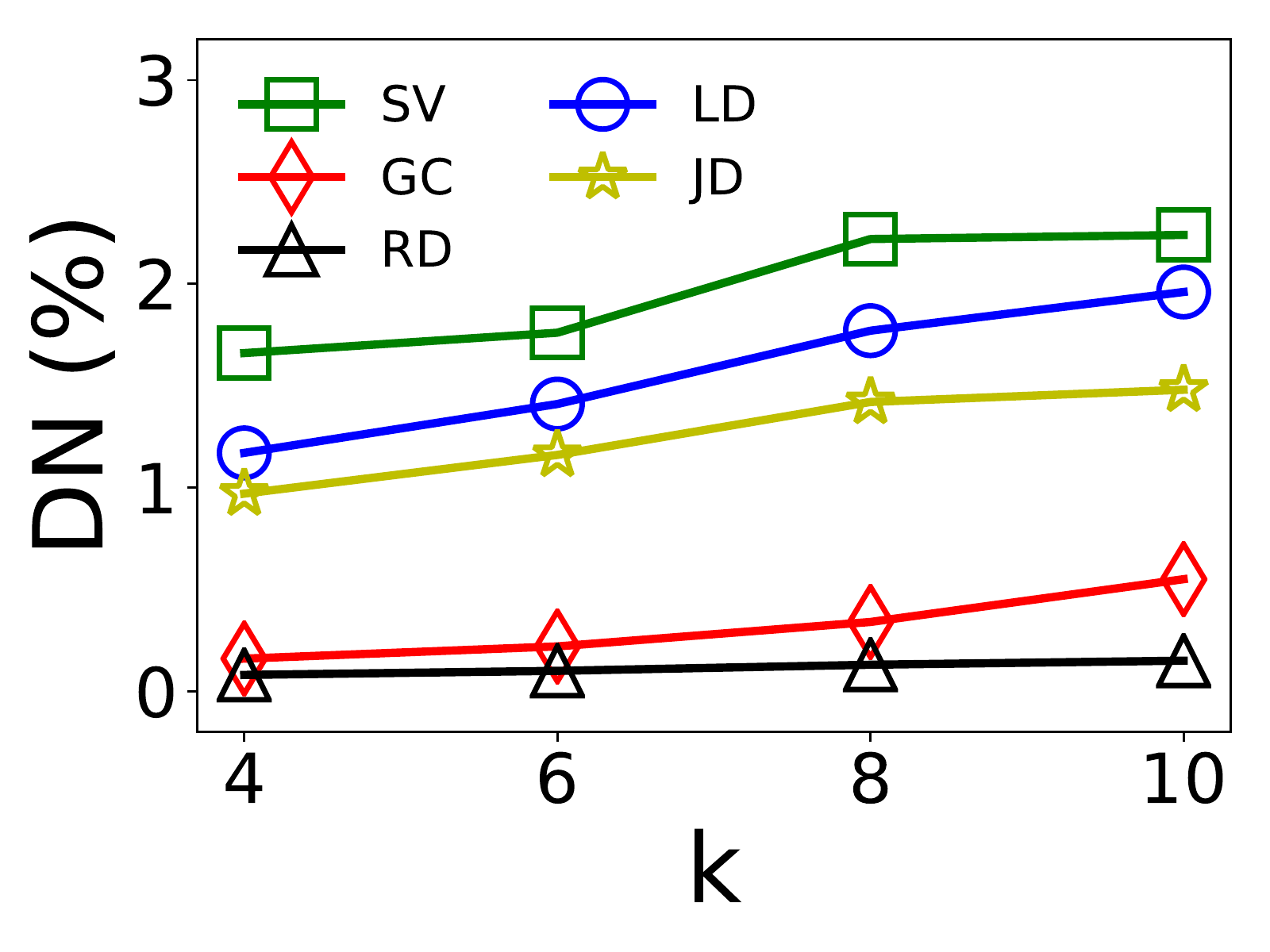}\label{fig:sv_fb_k}}
     \hspace{1mm}
     \subfloat[WS]{\includegraphics[width=0.24\textwidth]{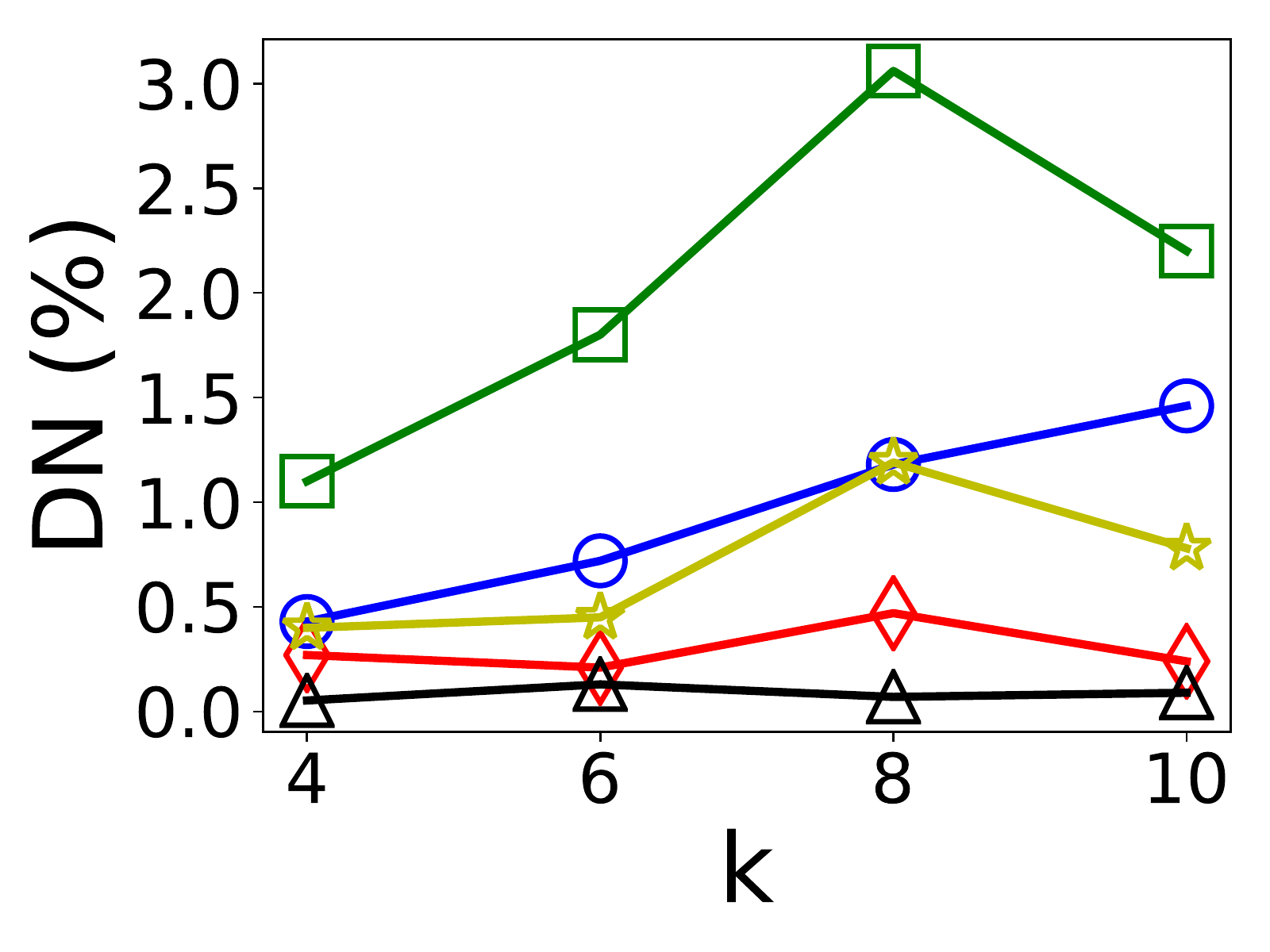}\label{fig:sv_ws_k}}
    \subfloat[FB]{\includegraphics[width=0.24\textwidth]{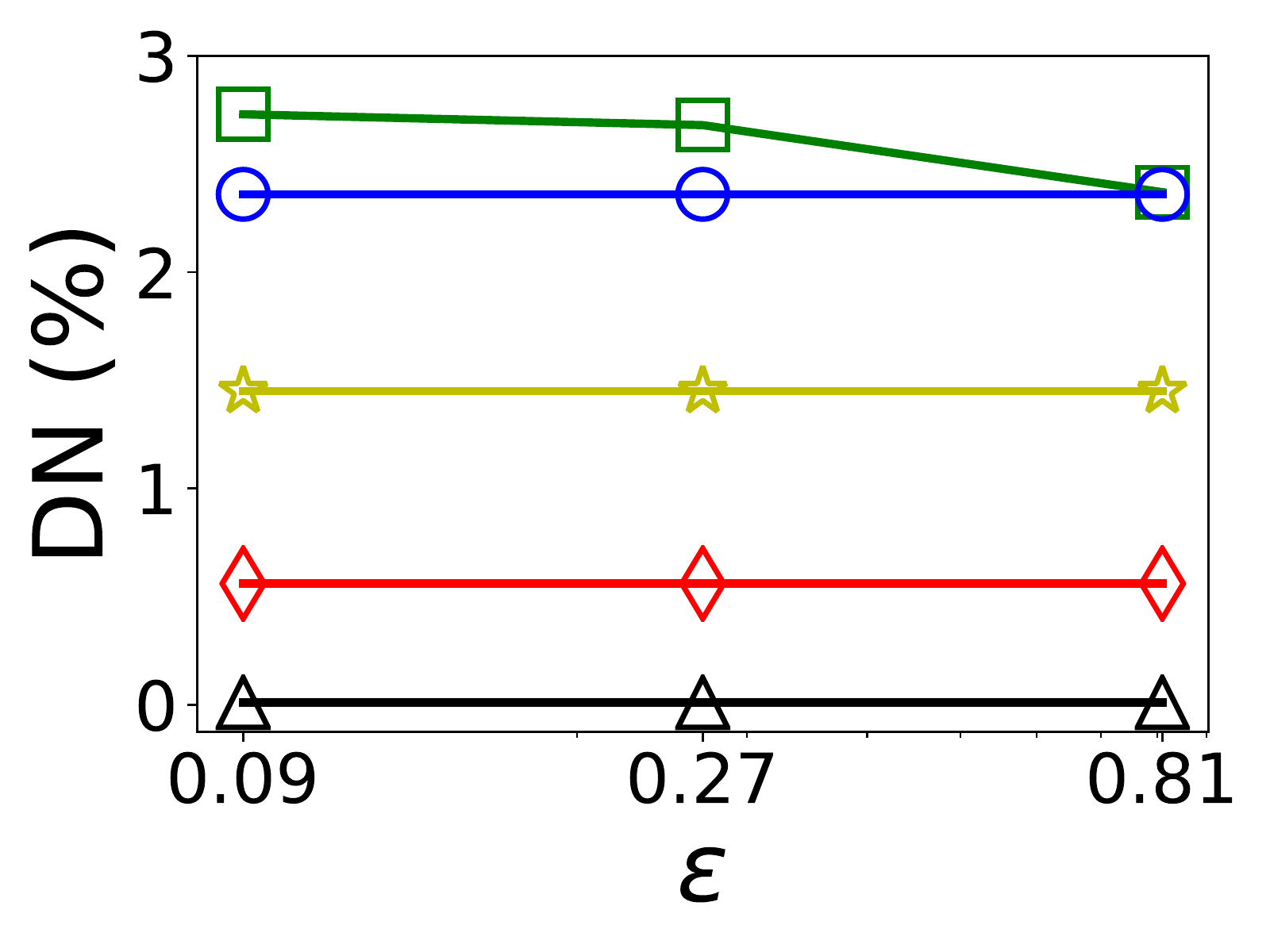}\label{fig:sv_fb_epsilon}}
    \hspace{1mm}
    \subfloat[WS]{\includegraphics[width=0.24\textwidth]{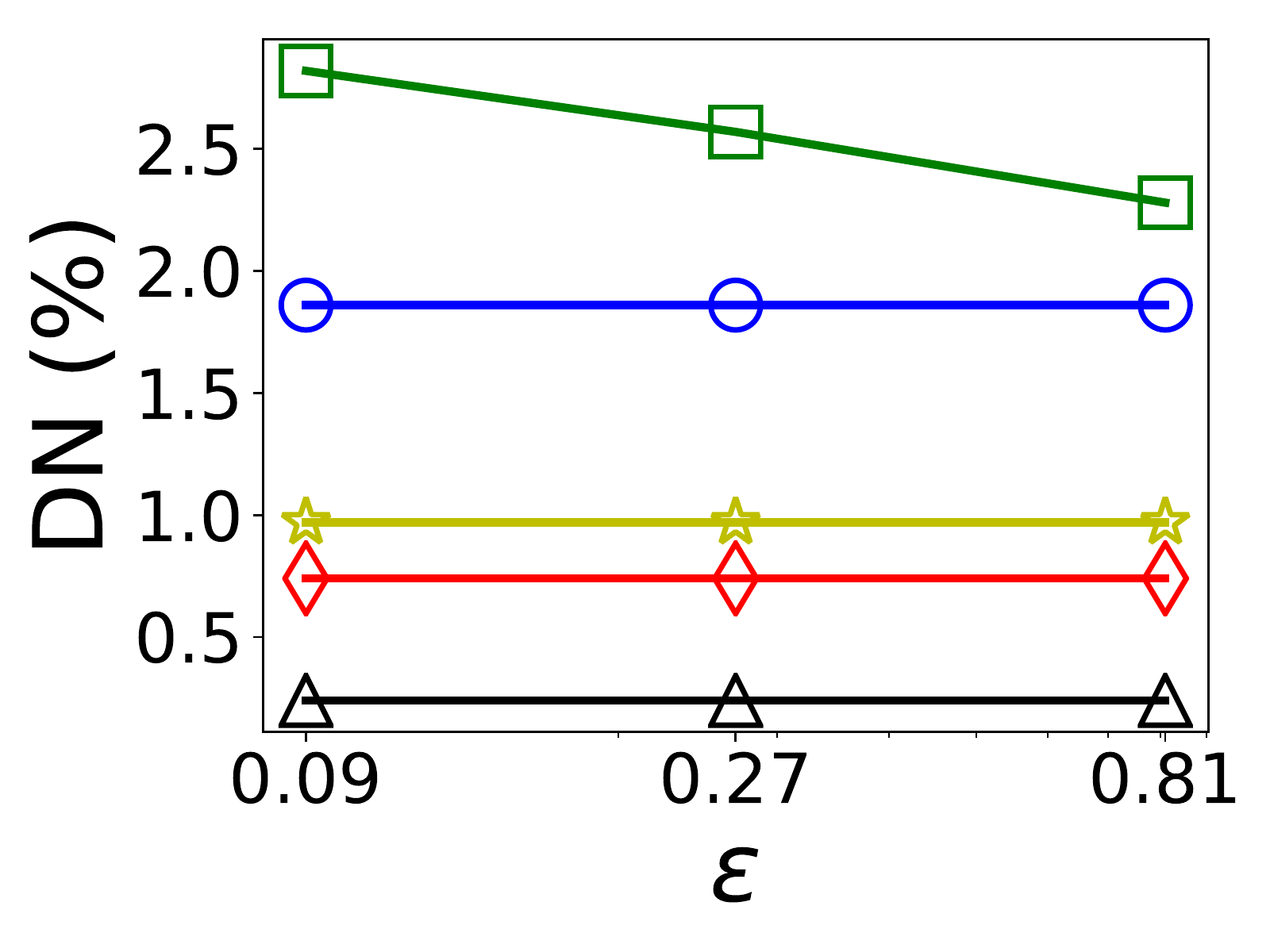}\label{fig:sv_ws_epsilon}}
    
    \vspace{-1mm}
    
    \caption{\textbf{K-core minimization (DN(\%)) for different algorithms varying (a-d) the number of edges in the budget; (e-f) the core parameter $k$; (g-h) and the sampling error $\epsilon$. Some combinations of experiments and datasets are omitted due to space limitations, but those results are consistent with the ones presented here. The Shapley Value based Cut (SV) algorithm outperforms the best baseline (LD) by up to 6 times. On the other hand, the Greedy approach (GC) achieves worse results than the baselines, with the exception of RD, in most of the settings. SV error increases smoothly with $\epsilon$ and LD becomes a good alternative for large values of $k$. \label{fig:SV_budget}}}
    \vspace{-2mm}
\end{figure*}

\subsection{Experimental Setup}

All the experiments were conducted on a $2.59$ GHz Intel Core i7-4720HQ machine with $16$ GB RAM running Windows 10. Algorithms were implemented in Java. The source-code of our implementations will be made open-source once this paper is accepted.

\textbf{Datasets:} The real datasets used in our experiments are available online and are mostly from SNAP\footnote{\url{https://snap.stanford.edu}}. The Human and Yeast datasets are available in \cite{moser2009mining}. In these datasets the nodes and the edges correspond to genes and interactions (protein-
protein and genetic interactions) respectively. The Facebook dataset is from \cite{viswanath2009evolution}. Table \ref{table:data_description} shows dataset statistics, including the  largest k-core (a.k.a. degeneracy). These are undirected and unweighted graphs from various applications: EE is from email communication; FB is an online social network, WS is a Web graph, DB is a collaboration network and CA is a product co-purchasing network. We also apply a random graph (ER) generated using the Erdos-Renyi model.

\textbf{Algorithms:} Our algorithm, \textit{Shapley Value Based Cut (SV)} is described in Section \ref{sec:shapley_algo}. 
Besides the Greedy Cut (GC) algorithm~\cite{zhu2018k} ( Section \ref{sec:algo_greedy}), we also consider three more baselines in our experiments. \textit{Low Jaccard Coefficient (JD)} removes the $k$ edges with lowest Jaccard coefficient. Similarly, \textit{Low-Degree (LD)} deletes $k$ edges for which adjacent vertices have the lowest degree. We also apply \textit{Random (RD)}, which simply deletes $k$ edges from the candidate set $\Gamma$ uniformly at random. Notice that while LD and JD are quite simple approaches for KCM, they often outperform GC.

\textbf{Quality evaluation metric:} We apply the percentage $DN(\%)$ of vertices from the initial graph $G$ that leave the $k$-core after the deletion of a set of edges $B$ (produced by a KCM algorithm):
\begin{equation} \label{eq:DN}
    DN(\%)=\frac{N_k(G)-N_k(G^B)}{N_k(G)}\times 100
\end{equation}

\textbf{Default parameters:} We set the candidate edge set $\Gamma$ to those edges ($M_k(G)$) between vertices in the k-core $C_k(G)$. Unless stated otherwise, the value of the approximation parameter for SV ($\epsilon$) is $0.05$ and the number samples applied is $\frac{\log{|\Gamma|}}{\epsilon^2}$ (see Theorem~\ref{thm:approx_SV}).

\subsection{Quality Evaluation}
\label{sec::effect_sv}

KCM algorithms are compared in terms of quality (DN(\%)) for varying budget ($b$), core value $k$, and the error of the sampling scheme applied by the SV algorithm ($\epsilon$). 

\textbf{Varying budget (b):} Figure \ref{fig:SV_budget} presents the k-core minimization results for $k\!=\!5$---similar results were found for $k\!=\!10$---using four different datasets. SV outperforms the best baseline by up to six times. This is due to the fact that our algorithm can capture strong dependencies among sets of edges that are effective at breaking the k-core structure. On the other hand, GC, which takes into account only marginal gains for individual edges, achieves worse results than simple baselines such as JD and LD. We also compare SV and the optimal algorithm in small graphs and show that SV produces near-optimal results (Section \ref{sec:sv_opt}).

\textbf{Varying core value (k):} We evaluate the impact of $k$ over quality for the algorithms using two datasets (FB and WS) in Figures \ref{fig:sv_fb_k} and \ref{fig:sv_ws_k}. The budget ($b$) is set to $400$. As in the previous experiments, SV outperforms the competing approaches. However, notice that the gap between LD (the best baseline) and SV decreases as $k$ increases. This is due to the fact that the number of samples decreases for higher $k$ as the number of candidate edge also decreases, but it can be mended by a smaller  $\epsilon$. Also, a larger $k$ will increase the level of dependency between candidate edges, which in turn makes it harder to isolate the impact of a single edge---e.g. independent edges are the easiest to evaluate. 
On the other hand, a large value of $k$ leads to a less stable k-core structure that can often be broken by the removal of edges with low-degree endpoints. 
LD is a good alternative for such extreme scenarios. Similar results were found for other datasets.

\textbf{Varying the sampling error ($\epsilon$):} The parameter $\epsilon$ controls the the sampling error of the SV algorithm according to Theorem \ref{thm:approx_SV}. We show the effect of $\epsilon$ over the quality results for FB and WS in Figures \ref{fig:sv_fb_epsilon} and \ref{fig:sv_ws_epsilon}. The values of $b$ and $k$ are set to $400$ and $12$ respectively. The performance of the competing algorithms do not depend on such parameter and thus remain constant. As expected, DN(\%) is inversely proportional to the value of $\epsilon$ for SV. The  trade-off between $\epsilon$ and the running time of our algorithm enables both accurate and efficient selection of edges for k-core minimization.

\subsection{Running Time} 
\label{sec:running_time}

Here, we evaluate the running time of the GC and SV algorithms. In particular, we are interested in measuring the performance gains due to the pruning strategies described in the Appendix. LD and JD do not achieve good quality results in general, as discussed in the previous section, thus we omit them from this evaluation. 

Running times for SV varying the sampling error ($\epsilon$) and the core parameter ($k$) using the FB dataset are given in Figures \ref{fig:fb_ep_time} and \ref{fig:fb_k_time}, respectively. Even for small error, the algorithm is able to process graphs with tens of thousands of vertices and millions of edges in, roughly, one minute. Running times decay as $k$ increases due to two factors: (1) the size of the $k$-core structure decreases (2) pruning gets boosted by a less stable core structure.

\begin{figure}[t]
\vspace{-4mm}
    \centering
    \subfloat[Varying $\epsilon$]{\includegraphics[width=0.22\textwidth]{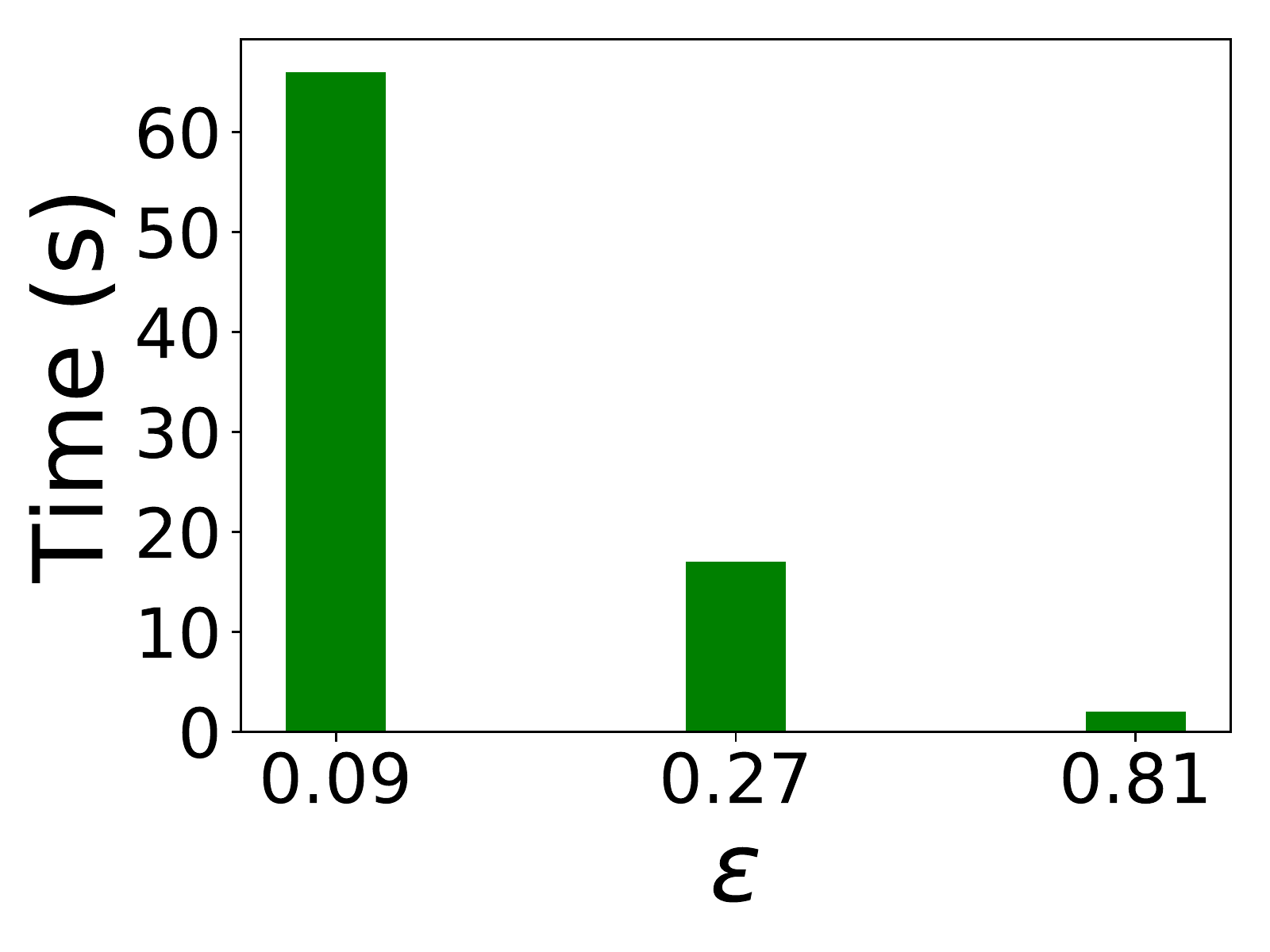}\label{fig:fb_ep_time}}
    \hspace{1mm}
    \subfloat[Varying $k$]{\includegraphics[width=0.22\textwidth]{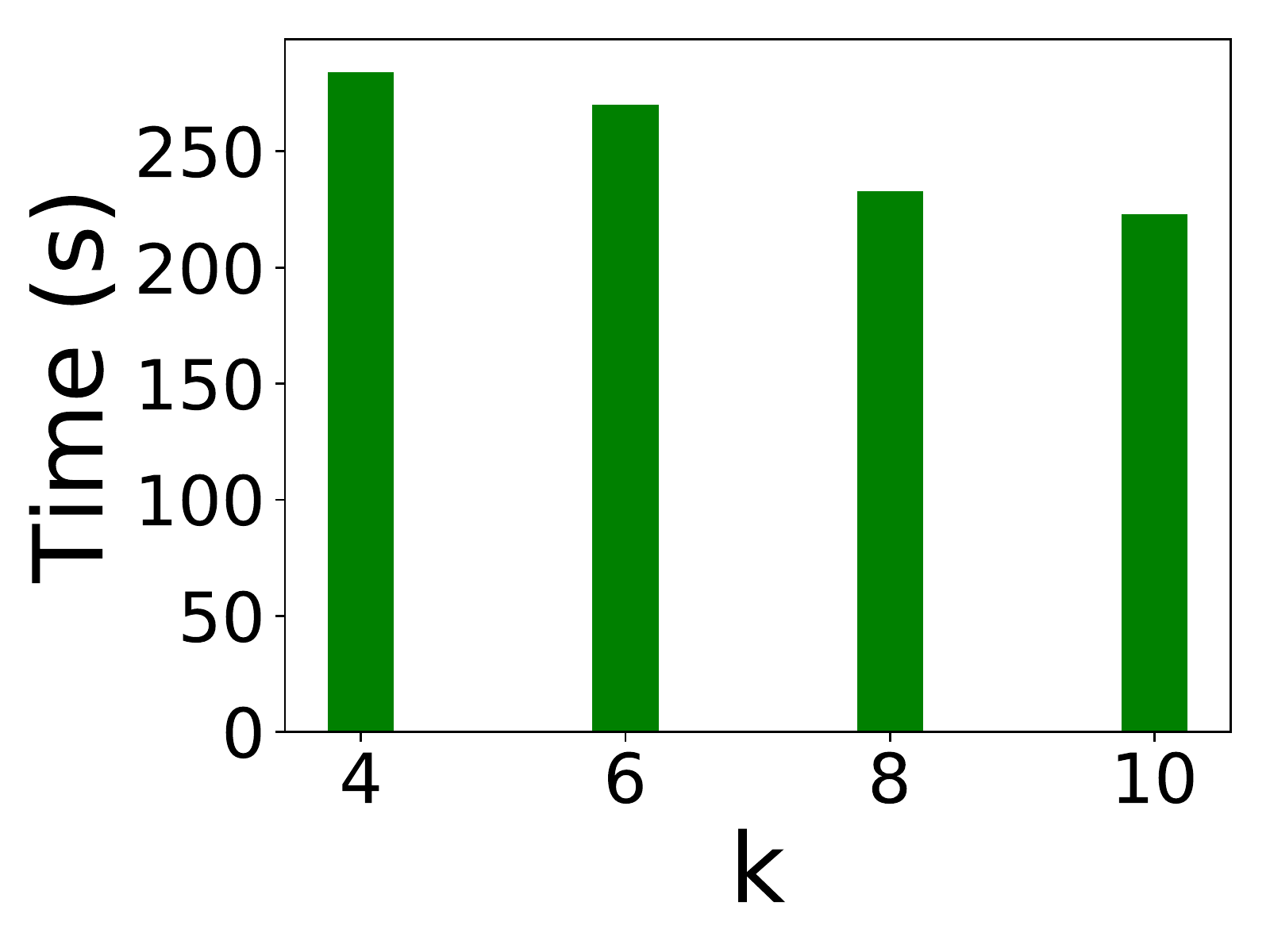}\label{fig:fb_k_time}}   
    \vspace{-1mm}
    \caption{Running times by SV using FB while varying (a) the sampling error $\epsilon$ and (b) the core parameter $k$; SV is efficient even for small values of sampling error and its running time decreases with $k$. \label{fig:gc_sv_time}}
    \vspace{-2mm}
\end{figure}


\begin{figure}[t]
\vspace{-4mm}
    \centering
    \subfloat[ $b=5$]{\includegraphics[width=0.245\textwidth, trim={1.5cm 1cm .5cm .45cm},clip]{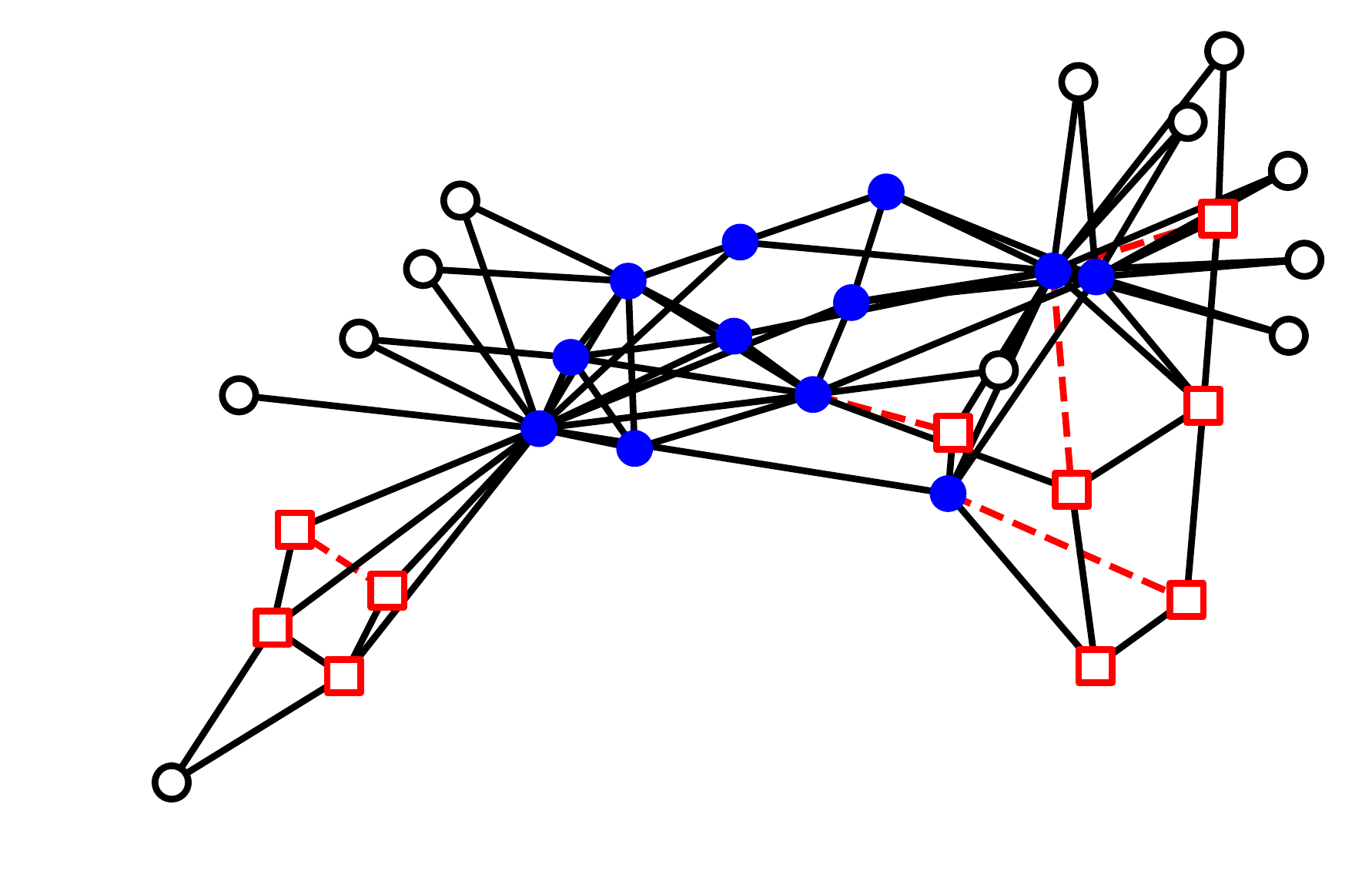}\label{fig:karate_b5_k3}}
    \subfloat[$b=10$]{\includegraphics[width=0.245\textwidth, trim={1.5cm 1cm .5cm .45cm},clip]{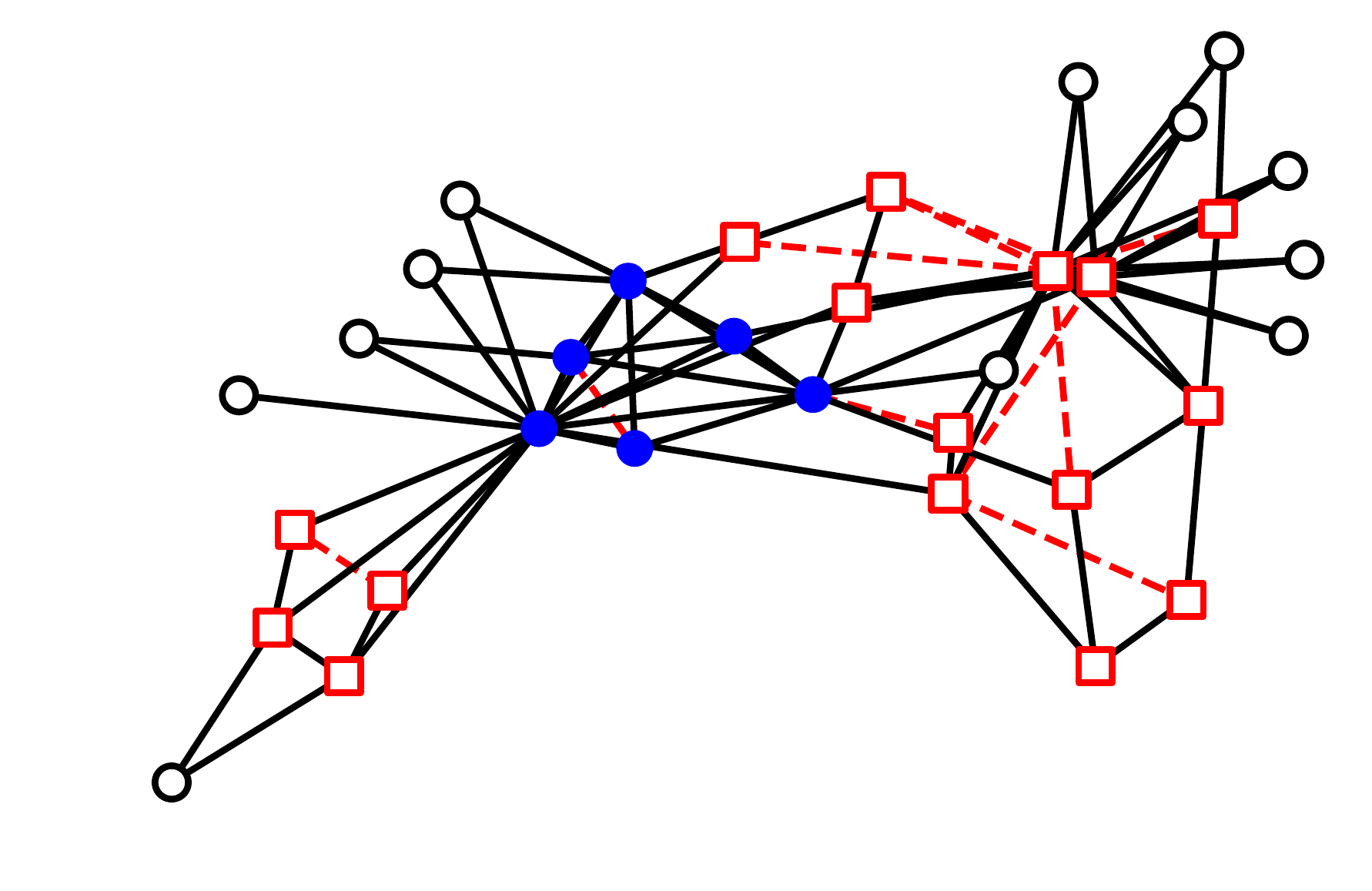}\label{fig:karate_b10_k3}}   
    \caption{K-core ($k=3$) minimization on the Zachary's Karate network: (a) $b=5$ and (b) $b=10$. Unfilled circle nodes are not in the $3$-core of the original network. After removal of $b$ dashed (red) edges, filled (blue) circle nodes remain in the $3$-core and unfilled (red) square nodes are removed from the $3$-core.  \label{fig:karate_k3}}
    \vspace{-2mm}
\end{figure}

\begin{figure*}[ht]
\vspace{-3mm}
    \centering
     \subfloat[DB]{\includegraphics[width=0.23\textwidth]{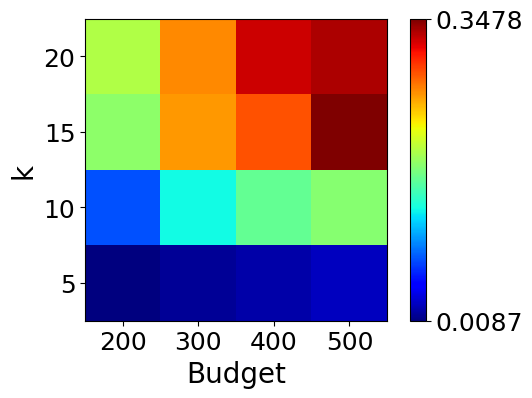}\label{fig:SV_k_b_DB}}
    \hspace{1mm}
    \subfloat[WS]{\includegraphics[width=0.23\textwidth]{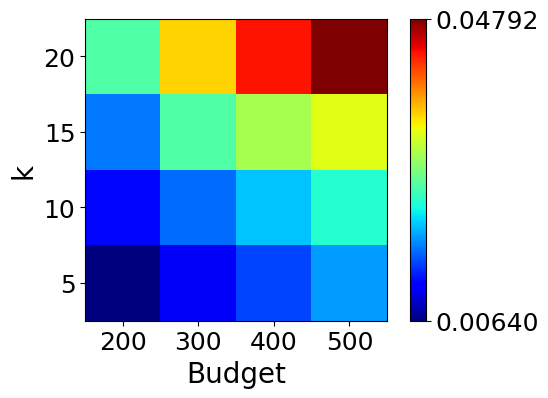}\label{fig:SV_k_b_WS}}
    \hspace{1mm} 
    \subfloat[FB]{\includegraphics[width=0.23\textwidth]{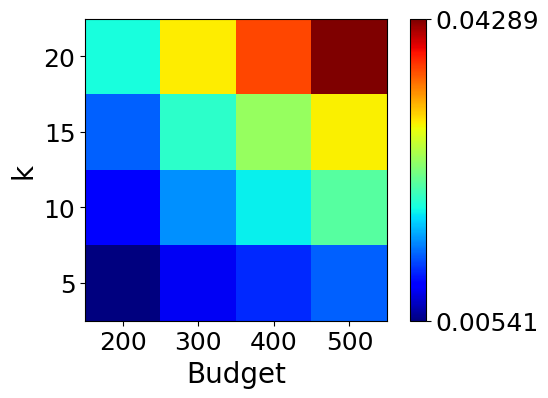}\label{fig:SV_k_b_FB}}
    \hspace{1mm}
    \subfloat[ER]{\includegraphics[width=0.23\textwidth]{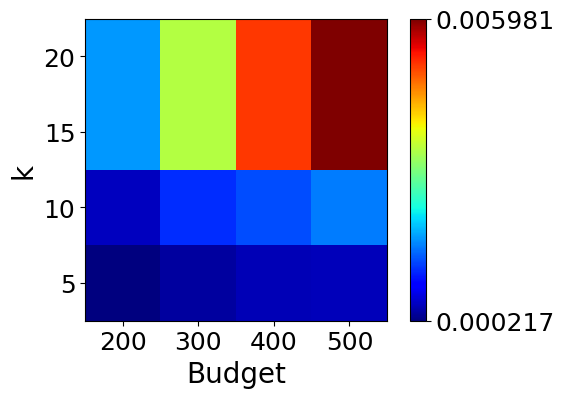}\label{fig:SV_k_b_ER}}
    \hspace{1mm}
    \vspace{-1mm}
    \caption{\textbf{Core resilience for four different networks: (a) DB (co-authorship), (b) WS (Webgraph), (c) FB (social), (d) ER (random). ER and DB are the most and least stable networks, respectively. Tipping points are found for ER and DB. \label{fig:SV_b_k_two}}}
\end{figure*}

\subsection{SV and the optimal algorithm} 
\label{sec:sv_opt}
In these experiments, we evaluate the approximation achieved by SV (Algorithm \ref{alg:SV}) compared to the optimal results using two small networks (Human and Yeast). The optimal set of $b\!=\!5$ and $b\!=\!10$ edges among a randomly chosen a set of $50$ edges is selected as the candidate set $\Gamma$ inside the $k$-core. An optimal solution is computed based on all possible sets with size $b$ in $\Gamma$. Table \ref{tab:sv_opt} shows the $DN(\%)$ produced by the optimal solution (OPT) and SV. Notice that the SV algorithm produces near-optimal results.
\begin{table}[t]
\centering
\begin{tabular}{|c| c | c | c | c |}
\hline
&\multicolumn{2}{c|}{\textbf{Human}} & \multicolumn{2}{c|}{\textbf{Yeast}}  \\
\hline
& $b=5$ & $b=10$ & $b=5$  & $b=10$\\
\hline
OPT & 2.88 & 3.24 & 11.16 & 12.05 \\
\hline
SV ($\epsilon =.1$) &  2.88 & 3.06 & 10.27 & 11.16 \\
\hline
SV ($\epsilon =.2$) &  2.8 & 3.06 & 8.48 & 10.71 \\
\hline
\end{tabular}

\caption{$DN(\%)$ by SV (approximate) and optimal algorithm using a small candidate set size ($|\Gamma|=50$), and $k=5$. \label{tab:sv_opt}}
\vspace{-9mm}
 \end{table}

\subsection{Application: $k$-core Resilience}
\label{sec:others}
We show how KCM can be applied to profile the resilience or stability of real networks. 
A profile provides a visualization of the resilience of the $k$-core structure of a network for different combinations of $k$ and budget. We apply $DN(\%)$ (Equation \ref{eq:DN}) as a measure of the percentage of the $k$-core removed by a certain amount of budget---relative to the immediately smaller budget value. 

Figure \ref{fig:SV_b_k_two} shows the results for co-authorship (DB), Web (WS), social network (FB) and a random (ER) graph. We also discuss profiles for Human and Yeast in the Appendix. Each cell corresponds to a given $k$-$b$ combination and the color of cell $(X,Y)$ shows the difference in $DN(\%)$ between $b\!=\!Y$ and $b\!=\!Y\!-\!100$ for $k\!=\!X$. 
As colors are relative, we also show the range of values associated to the the color scheme. We summarize our main findings as follows:

\textbf{Stability:} ER (Figure \ref{fig:SV_k_b_ER}) is the most stable graph, as can be noticed by the range of values in the profile. The majority of nodes in ER are in the $19$-core. DB (Figure \ref{fig:SV_k_b_DB}) is the least stable, but only when $k\!>\!5$, which is due to its large number of small cliques. The high-core structure of DB is quite unstable, with less than $1$\% of the network in the $20$-core structure after the removal of $500$ edges.

\textbf{Tipping points:} We also look at large effects of edge removals within small variations in budget---for a fixed value of $k$. Such a behavior is not noticed for FB and WS (Figures \ref{fig:SV_k_b_WS} and \ref{fig:SV_k_b_FB}, respectively), for which profiles are quite smooth. This is mostly due to the presence of fringe nodes at different levels of $k$-core structure. On the other hand, ER produced the most prominent tipping points ($k\!=\!15$ and $k\!=\!20$). This pattern is also found for DB.

\subsection{$K$-core Minimization on the Karate Network}
In Figure \ref{fig:karate_k3}, we demonstrate the application of our algorithm for KCM  using the popular Zachary's Karate network with two different budget settings, $b=5$ and $b=10$, and $k$ fixed to $3$. Unfilled circles are nodes initially out of the $3$-core. The dashed (red) edges are removed by our algorithm---often connecting fringe nodes. Filled (blue) circles and unfilled (red) squares represent nodes that remain and are removed from the $3$-core, respectively, after edge removals.  

\section{Previous Work}
\label{sec:prev_work}

\noindent
\textbf{ $K$-core computation and applications: } 
A $k$-core decomposition algorithm was first introduced by Seidman \cite{seidman1983network}. 
A more efficient solution---with time complexity $O(|E|)$---was presented by Batagelj et al. \cite{batagelj2011fast} and its distributed version was proposed in \cite{montresor2013distributed}. Sariyuce et al. \cite{sariyuce2013streaming} and Bonchi et al. \cite{bonnet2016parameterized} proposed algorithms $k$-core decomposition in streaming data and uncertain graphs respectively. $K$-cores are often applied in the analysis and visualization of large scale complex networks \cite{alvarez2006large}. Other applications include clustering and community detection \cite{giatsidis2014corecluster}, characterizing the Internet topology \cite{carmi2007model}, and analyzing the structure of software systems \cite{zhang2010using}. In social networks, $k$-cores are usually associated with models for user engagement. Bhawalkar et al. \cite{bhawalkar2015preventing} and Chitnis et al. \cite{chitnis2013preventing} studied the problem of increasing the size of $k$-core by anchoring a few vertices initially outside of the $k$-core. Malliaros et al. \cite{malliaros2013stay} investigated user engagement dynamics via $k$-core decomposition. 
\textbf{Network resilience/robustness:} Understanding the behavior of a complex system (e.g. the Internet, the power grid) under different types of attacks and failures has been a popular topic of study in network science \cite{callaway2000network,albert2004structural,cohen2000resilience}. This line of work is mostly focused on non-trivial properties of network models, such as critical thresholds and phase transitions, assuming random or simple targeted modifications. Najjar et al. \cite{najjar1990network} and Smith et al. \cite{smith2011network} apply graph theory to evaluate the resilience of computer systems, specially communication networks. An overview of different graph metrics for assessing robustness/resilience is given by \cite{ellens2013graph}. Malliaros et al. \cite{malliaros2012fast} proposed an efficient algorithm for computing network robustness based on spectral graph theory. The appropriate model for assessing network resilience and robustness depends on the application scenario and comparing different such models is not the focus of our work.

\textbf{Stability/resilience of $k$-core: } Adiga et al. \cite{adiga2013robust} studied the stability of high cores in noisy networks. Laishram et al. \cite{Laishram2018} recently introduced a notion of resilience in terms of the stability of $k$-cores against deletion of random nodes/edges. If the rank correlation of core numbers before and after the removal is high, the network is core-resilient. They also provided an algorithm to increase resilience via edge addition. Notice that this is different from our problem, as we search for edges that can destroy the stability of the $k$-core. Another related paper is the work by Zhang et al. \cite{zhang2017finding}. Their goal is to find $b$ vertices such that their deletion reduces the $k$-core maximally. The $k$-core minimization problem via edge deletion has been recently proposed by Zhu et al. \cite{zhu2018k}. However, here we provide stronger inapproximability results and a more effective algorithm for the problem, as shown in our experiments.

\textbf{Shapley Value (SV) and combinatorial problems:}  A Shapley value based algorithm was previously introduced for influence maximization (IM) \cite{narayanam2011shapley}. However, IM can be approximated within a constant-factor by a simple greedy algorithm due to the submodular property \cite{kempe2003maximizing}. In this paper, we use Shapley value to account for the joint effect of multiple edges in the solution of the KCM problem, for which we have shown stronger inapproximability results.

\textbf{Shapley Value (SV) estimation via sampling:}  Sampling techniques have been applied for the efficient computation of SV's \cite{castro2009polynomial, maleki2013bounding}. Castro et al. \cite{castro2009polynomial} introduced SV sampling 
in the context of symmetric and non-symmetric voting games. Maleki et al. \cite{maleki2013bounding} provided analyses for stratified  sampling specially when the marginal contributions of players are similar. Our sampling results are specific for the KCM problem and analyze the effect of sampling on the top $k$ nodes with highest Shapley values (Corollary \ref{cor:anyset}).

\textbf{Other network modification problems: }
A set of network modification problems were introduced by Paik et al.~\cite{paik1995}. Recent work~\cite{lin2015, dilkina2011} addressed the shortest path distance optimization problem via improving edge or node weights on undirected graphs. A node version of the problem has also been studied \cite{dilkina2011,medya2018noticeable}. 
Another related problem is to optimize node centrality by adding edges \cite{crescenzi2015,ishakian2012framework,medya2018group}. Boosting or containing diffusion processes in networks were studied under different well-known diffusion models such as Linear Threshold \cite{Khalil2014} and Independent Cascade \cite{kimura2008minimizing}.

\section{Conclusion}
We have studied the $k$-core minimization (KCM) problem, which consists of finding a set of edges, removal of which minimizes the size of the $k$-core structure. KCM was shown to be NP-hard, even to approximate within any constant when $k\!\geq\!3$. The problem is also not fixed-parameter tractable, meaning it cannot be solved efficiently even if the number of edges deleted is small. Given such inapproximability results, we have proposed  an efficient randomized heuristic based on Shapley value to account for the interdependence in the impact of candidate edges. For the sake of comparison, we also evaluate a simpler greedy baseline, which cannot assess such strong dependencies in the effects of edge deletions. 

We have evaluated the algorithms using several real graphs and shown that our Shapley value based approach outperforms competing solutions in terms of quality. The proposed algorithm is also efficient, enabling its  application to graphs with hundreds of thousands of vertices and millions of edges in time in the order of minutes using a desktop PC. We have also illustrated how KCM can be used for profiling the resilience of networks to edge deletions.

\section*{Acknowledgment}
The authors would like to thank Neeraj Kumar for helpful discussions.



\bibliographystyle{ACM-Reference-Format}  
\bibliography{AAMAS20}  
\section{Appendix}

\subsection{Proof for Theorem \ref{thm:np_hard_k_1_2}}

\begin{proof}
First, we sketch the proof for $k=1$. 
Consider an instance of the NP-hard 2-MINSAT \cite{kohli1994minimum} problem which is defined by a set $U=\{u_1,u_2,...,u_{m'}\}$ of $m'$ variables and collection $C'=\{c_1,c_2,...,c_{n'}\}$ of $n'$ clauses. Each clause $c\in C'$ has two literals ($|c|=2$). So, each $c_i\in C'$ is of the form $z_{i1}\vee z_{i2}$ where $z_{ij}$ is a literal and is either a variable or its negation. The problem is to decide whether there exists a truth assignment in $U$ that satisfies no more than $n^* < n'$ clauses in $C$. To define a corresponding KCM instance, we construct the graph $G'$ as follows. We create a set of $n'$ vertices $X_c =\{v_i|\ c_i\in C'\}$. For each variable $u_i\in U$, we create two vertices: one for the variable ($w_{i1}$) and another for its negation ($w_{i2}$). Thus, a total of $2m'$ vertices, $Y_u=\{w_{11},w_{12},w_{21},w_{22}$ $, \ldots, w_{m'1},w_{m'2}\}$ are produced. Moreover, whenever the literal $u_i\vee \bar{u}_j\in c_t$, we add two edges, $(v_{t},w_{i1})$ and $(v_{t},w_{j2})$ to $G'$.

For $k=1$, KCM aims to maximize the number of isolated vertices ($0$-core, $k\!=\!0$) via removing $b$ edges. An edge in the KCM instance ($K_I$) is a vertex $v_i$ in $G'$. Each vertex is connected to exactly two vertices (end points of the edge in $K_I$) in $Y_u$. Satisfying a clause is equivalent to removing the corresponding vertex (deleting the edge in $K_I$) from $G'$. A vertex in $Y_u$ will be isolated when all associated clauses (or vertices) in $X_c$ are satisfied (removed). If there is a truth assignment which satisfies no more than $b=n^*$ clauses in 2-MINSAT, that implies $m'$ vertices can be isolated in $G'$ by removing $\leq b$ vertices (or deleting $\leq b$ edges in KCM). If there is none, then $m'$ vertices cannot be isolated by removing $\leq b$ edges in KCM.

To prove for $k=2$, we can transform the  $k=1$ version of KCM to the $k=2$ one (transformation is similar to the one in \cite{zhang2017finding}). 
\end{proof}


\subsection{Proof for Theorem \ref{thm: param_approx} }
\begin{proof}
We sketch the proof for $k=3$. A similar construction can be applied for $k>3$. Consider an instance of the $W[2]$-hard Set Cover \cite{bonnet2016parameterized} problem, defined by a collection of subsets $S=\{S_{1},S_{2},...,S_{m}\}$ from a universal set of items $U=\{ u_{1},u_{2},...,u_{n} \}$. The problem is to decide whether there exist $b$ subsets whose union is $U$. We define a corresponding KCM instance (graph $G$) as follows. 
 
 For each $S_i \in S$ we create two cycles: one of $n$ vertices $X_{i,1},X_{i,2},$\\$\cdots,X_{i,n}$ in $V$ with edges $(X_{i,1},X_{i,2}), (X_{i,2},X_{i,3}),\cdots,(X_{i,n},X_{i,1})$. Another cycle of $n$ vertices $W_{i,1}$ to $W_{i,n}$ and obviously with $n$ edges $(W_{i,1},W_{i,2}), (W_{i,2},W_{i,3})\cdots,(W_{i,n},W_{i,1})$ is also added along with $n$ more edges $(W_{i,j},X_{i,j})$ where $j=1,2,\cdots,n$ for every $i$. Moreover, for each $u_j \in U$, we create a cycle of $m$ vertices $Y_{j,1}, Y_{j,2},\cdots, Y_{j,m}$ with the following edges $(Y_{j,1},Y_{j,2}), \cdots,(Y_{j,m-1},Y_{j,m}), (Y_{j,m},Y_{j,1})$. 
 We add $5$ vertices $Z_{j,1}$ to $Z_{j,5}$ with cliques of four vertices $Z_{j,2}$ to $Z_{j,5}$ and two more edges: $(Z_{j,1},Z_{j,2}),(Z_{j,1},Z_{j,5})$. 
 Furthermore, edges $(X_{i,j},Y_{j,i})$ will be added to $E$ if $u_j\in S_i$. 
 Moreover, edges $(Y_{j,i},Z_{j,1})$ will be added to $E$ if $u_j\notin S_i$. Clearly the reduction is in FPT. The candidate set, $\Gamma = \{ (W_{i,1},W_{i,2})| \forall{i=1,2,...,m}\}$.  Fig. \ref{fig:hardness_ex2} illustrates our construction for sets $S_1=\{u_1\},S_2=\{u_1,u_2,u_4\},$ and $S_3=\{u_3\}$.
 
Initially all nodes are in the $3$-core. We claim that a set $S'\subset S$, with $|S'|\leq b$, is a cover iff $f_3(B)\!=\! 2bn\!+\!n(m\!+\!1)$ where  $B\!=\!\{(W_{i,1},W_{i,2})| S_i \in S' \}$. 
For any $i$, if $(W_{i,1},W_{i,2})$ is removed, the $2n$ nodes $\{W_{i,1},\cdots,W_{i,n}\}$ and $\{X_{i,1},\cdots,X_{i,n}\}$ will go to the $2$-core (note that the trivial case is removed where a subset will contain all the elements; thus there will be one $X_{i,t}$ for some $t$ with exactly degree 3). If $u_j\in S_i$, then $m+1$ nodes  $\{Y_{j,1},...,X_{j,m}\}$ and $Z_{j,1}$ go to $2$-core after $(W_{i,1},W_{i,2})$ is removed. If $S'$ is a set cover, all the $u_j$s will be in some $S_i\in S'$ and $n(m+1)$ nodes (all the $Y_{j,i}$s and $Z_{j,1}$s) will go into $2$-core; so $f_3(B)\!=\!2bn\!+\!n(m+1)$. On the other hand, assume that $f_3(B)\!=\!2bn\!+\!n(m\!+\!1)$ after removing edges in $B= \{(W_{i,1},W_{i,2})| S_i \in S' \}$. The only way to have $m+1$ nodes removed from corresponding $u_j$ is if $u_j \in X_i$ and $X_i\in X$. Thus, $n(m+1)$ nodes (all the $Y_{j,i}$s and $Z_{j,1}$s) along with $2bn$ nodes ($X_{i,j}$s and $W_{i,j}$s, where $S_i\in S'$) will be removed, making $S'$ a set cover. 
\end{proof}

\begin{figure}[t]
    \centering
    {\includegraphics[width=0.4\textwidth]{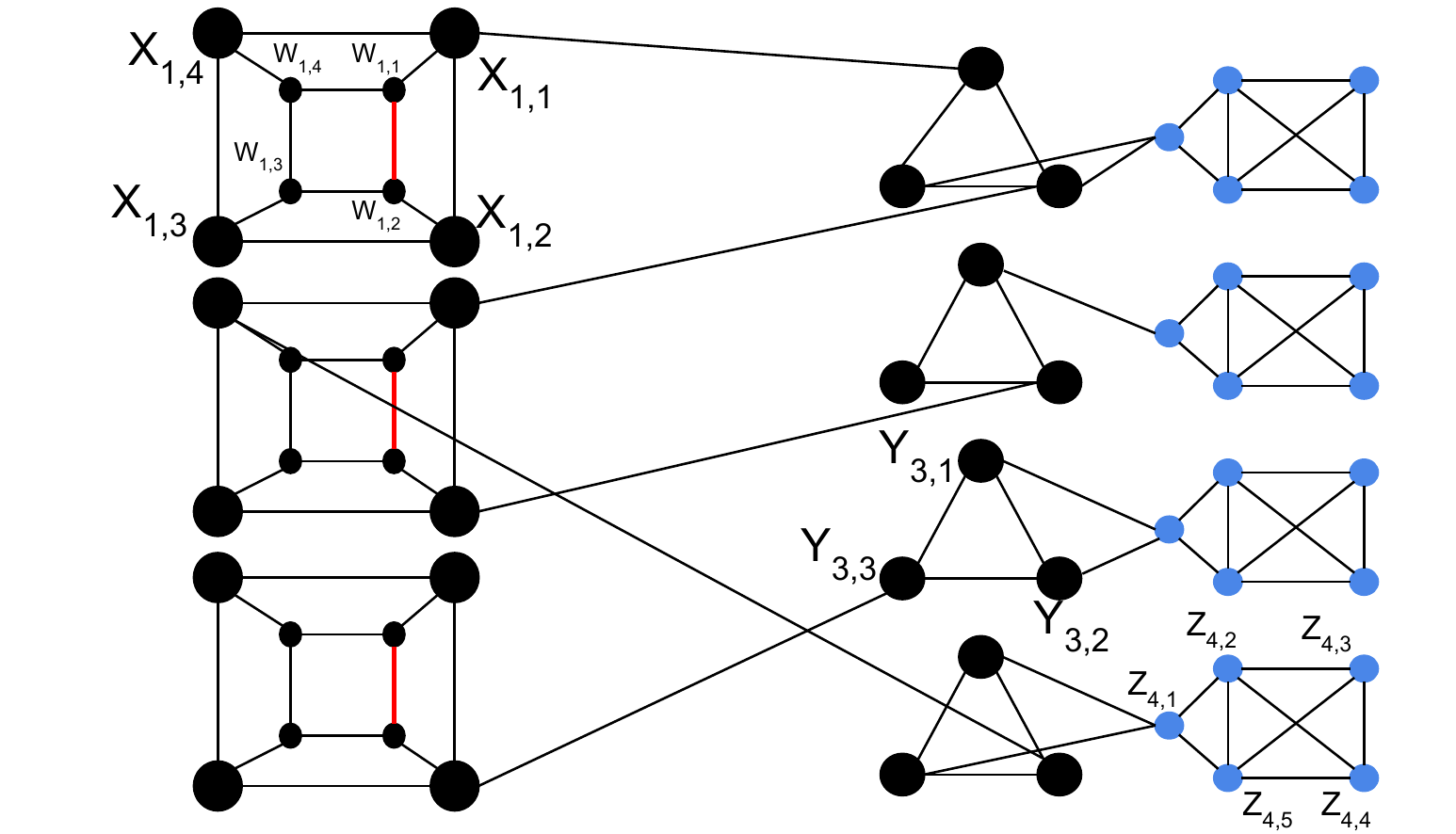}}
    \caption{Example construction for parameterized hardness from Set Cover where  $U=\{u_1,u_2,u_3,u_4\}, S=\{S_1,S_2,S_3\}, S_1=\{u_1\},S_2=\{u_1,u_2,u_4\},$ and $S_3=\{u_3\}$. \label{fig:hardness_ex2}}
\end{figure}

\subsection{Algorithm \ref{alg:computeVS}} \label{sec:algo_computevs}
This procedure computes the vulnerable set---i.e., the set of nodes that will leave the $k$-core upon deletion of the edge $e$ from $\mathscr{G}$. The size of the set is essentially the marginal gain of deleting $e$. If $e\!=\!(u,v)$ is deleted, $u$ will be removed iff $d(\mathscr{G},u)\!=\!k$ (the same for $v$). This triggers a cascade of node removals from the $k$-core (with the associated edges). Let $vul(w)$ be the set of nodes already removed from $\mathscr{G}$ that are neighbours of node $w$. We observe that $w$ will be removed if $d(\mathscr{G},w)- |vul(w)|<k$. Note that the procedure is similar to Algorithm 2 (\textit{LocalUpdate}), having $O(M_k+N_k)$ running time.

\begin{algorithm}[h]
\small
\caption{computeVS}
\label{alg:computeVS}
\KwInput{$e=(u,v),\mathscr{G},k$}
\KwOutput{$X$}
\If{$d(\mathscr{G},u)= \mathcal{E}_\mathscr{G}(u)$}{
 Queue $S\leftarrow S\cup \{u\}$,  $X\leftarrow X\cup \{u\}$\\
}
\If{$d(\mathscr{G},v)= \mathcal{E}_\mathscr{G}(v)$}{
 Queue $S\leftarrow S\cup \{v\}$, $X\leftarrow X\cup \{v\}$\\
}
\While{$S \not = \emptyset$ } {
    Remove $y$ form $S$ \\
    \For{$w \in N(y)$}{
     $vul(w)\leftarrow \{z|z\in N(w)\cap X\}$ \\
    \If{$w \not \in X \And d(\mathscr{G},w) - |vul(w)|<k $}{
    Add $w$ to $X$, $S$
    }
    }
}
\textbf{return} $X$
\end{algorithm}

\subsection{Local Update (Algorithm \ref{alg:Local_update})}  \label{sec:algo_local}
After the removal of the edge $e^*$ in each step, the current graph $\mathscr{G}$ is updated (step $9$). Recomputing the $k$ cores in $\mathscr{G}$ would take $O(M_k)$ time. Instead, a more efficient approach is to update only the affected region after deleting the $e^*$. If an edge $e^*=(u,v)$ is deleted, $u$ will be removed if $d(\mathscr{G},u)=k$ (the same for $v$). This triggers a cascade of node removals (with the associated edges). Let $vul(w)$ be a set of nodes already removed from $\mathscr{G}$ that are neighbours of node $w$. We observe that $w$ will be removed if $d(\mathscr{G},w)- |vul(w)|<k$.

\begin{algorithm}[t]
\small
\caption{LocalUpdate}
\label{alg:Local_update}
\KwInput{$e=(u,v),\mathscr{G},k$}
 Remove $(u,v)$, update $d(\mathscr{G},u),d(\mathscr{G},v)$, $X\leftarrow \Phi$, $Y\leftarrow \Phi$\\
    \If{$d(\mathscr{G},u)<k$}{
     Queue $Y\leftarrow Y\cup \{u\}$, $X\leftarrow X\cup \{u\}$
    }
    \If{$d(\mathscr{G},v)<k$}{
     Queue $Y\leftarrow Y\cup \{v\}$, $X\leftarrow X\cup \{v\}$
    }
\While{$Y \not = \emptyset$} {
    Remove $y$ form $Y$ 
       
        \For{$w \in N(y)$}{
                 $vul(w)\leftarrow \{z|z\in N(w)\cap X\}$ \\
                \If{$w \not \in X \And d(\mathscr{G},w) - |vul(w)|<k $}{
                 Add $w$ to $X$, $S$
                }

        }
    \If{$d(\mathscr{G},y) < k$}{
    Remove $y$ from $\mathscr{G}$
    }
}
\end{algorithm}
\subsection{Optimizations for GC and SV }
\label{sec:opt}
Here, we discuss optimizations for the Greedy (GC) and Shapley Value based (SV) algorithms. We propose a general pruning technique to speed up both Algorithms \ref{alg:GC} and \ref{alg:SV} (GC and SV). For GC, in each step, all the candidate edges are evaluated (step $3$). \textit{How can we reduce the number of evaluations in a single step?} In SV, in a single permutation, marginal gains are computed for all the candidate edges (step $5$). \textit{How can we skip edges that have $0$ marginal gain?}. We answer these questions by introducing the concept of \textit{edge dominance}. Let $Z(e,\mathscr{G})$ be the set of vertices that would be removed if $e$ is deleted from $\mathscr{G}$ due to the $k$-core constraint. If $e'=(u,v)$ has one of the end points $u$ or $v$ in $Z(e,\mathscr{G})$, then $e'$ is dominated by $e$.

\begin{observation} \label{obs:dominance}
If $e'$ is dominated by $e$, then $Z(e',\mathscr{G}) \subseteq Z(e,\mathscr{G})$.
\end{observation}

In Algorithm \ref{alg:GC} (GC), while evaluating each edge in the candidate set (step $3$) if $e'$ comes after $e$, we skip the evaluation of $e'$, as $|X_e| \geq |X_{e'}|$ (Obs. \ref{obs:dominance}). 
In Algorithm \ref{alg:SV} (SV), while computing the marginal gain of each edge in a coalition for a particular permutation $\pi$, assume that $e'$ appear after $e$.  As $e \in P_{e'}(\pi)$ and using Observation \ref{obs:dominance}, $\mathscr{V}(P_e (\pi)\cup \{e\}) - \mathscr{V}(P_e (\pi))) =0$. Thus, the computation of the marginal gain of $e'$ can be skipped. We evaluate the performance gains due to pruning in our experimental evaluation.

In Figures \ref{fig:gc_time} and \ref{fig:sv_time}, we look further into the effect of pruning for GC and SV by comparing versions of the algorithms with and without pruning using three datasets. GC becomes one order of magnitude faster using pruning. Gains for SV are lower but still significant (up to 50\%). We found  in other experiments that the impact of pruning for SV increases with the budget, which is due to the larger number of permutations to be considered by the algorithm.
 
 \begin{figure}[t]
\vspace{-2mm}
    \centering
     \subfloat[Pruning, GC]{\includegraphics[width=0.22\textwidth]{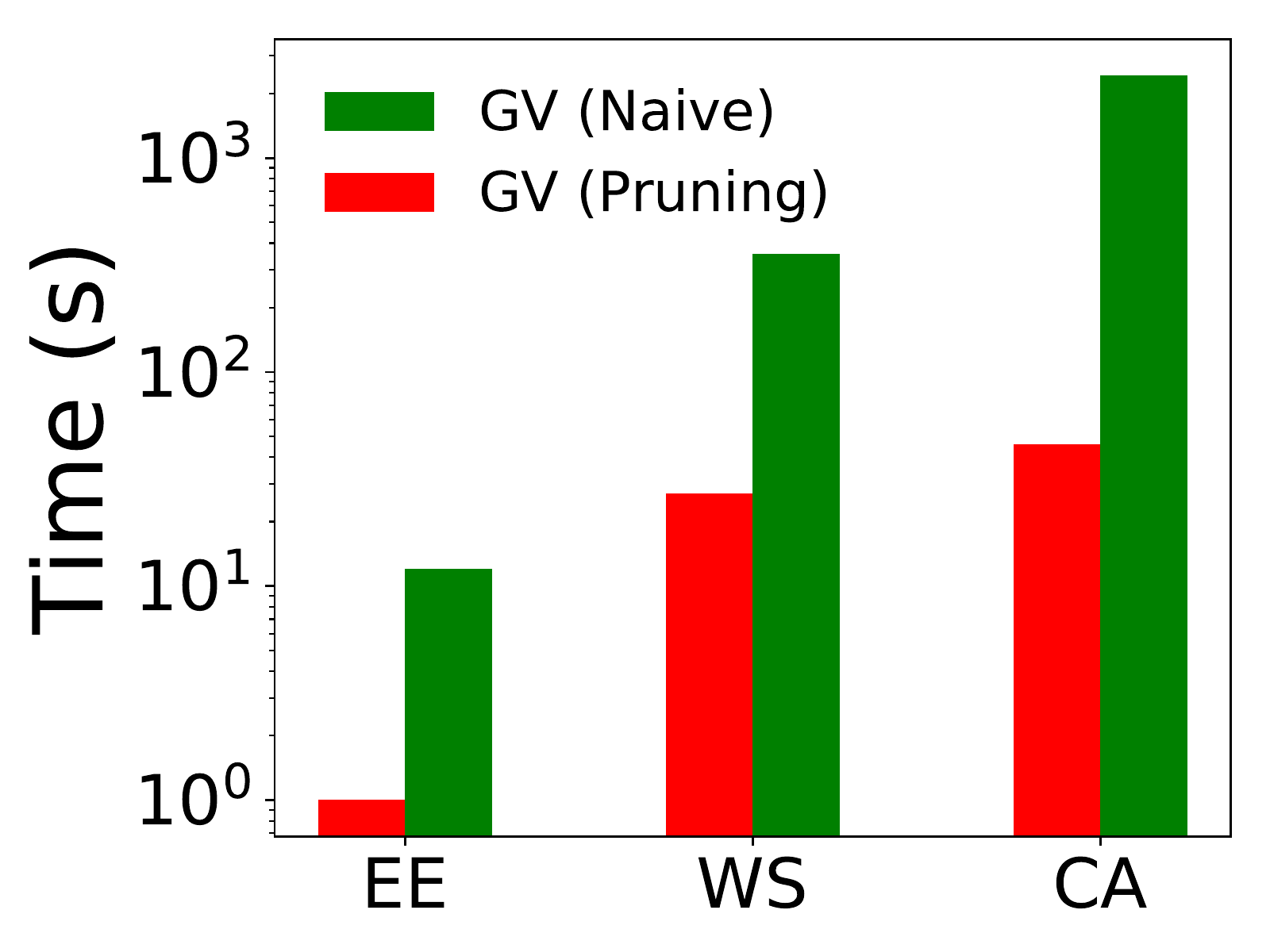}\label{fig:gc_time}}
    \hspace{1mm}
       \subfloat[Pruning, SV]{\includegraphics[width=0.22\textwidth]{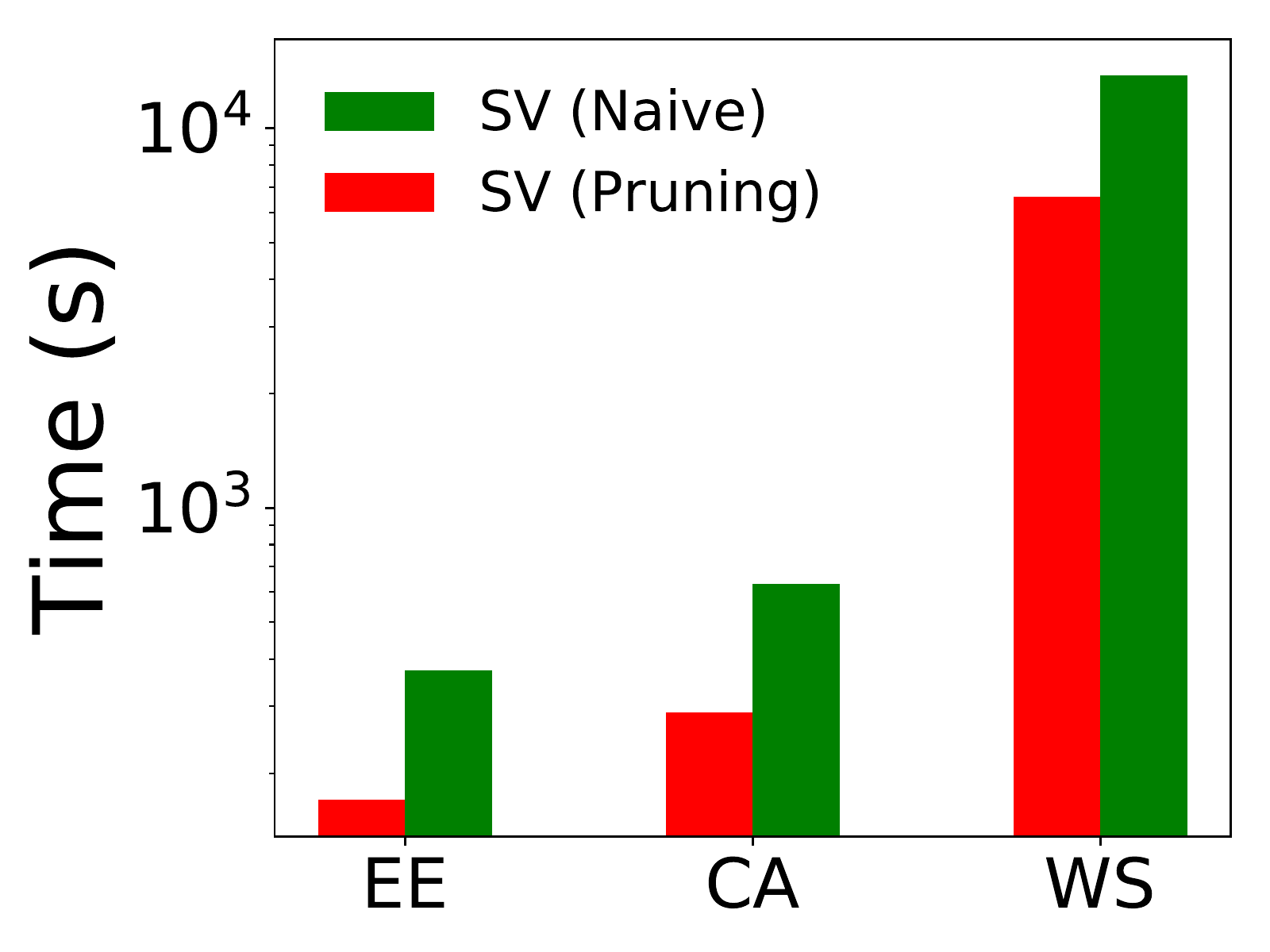}\label{fig:sv_time}}
    \hspace{1mm}
    \caption{Impact of pruning for GC and SV algorithms using three datasets: GC is up to one order of magnitude faster with pruning, while SV is up to 50\% faster. \label{fig:gc_sv_time_1}}
    \vspace{-2mm}
\end{figure}

\subsection{K-core Resilience: Human vs Yeast} 
\label{sec:human_vs_yeast}
K-cores have been previously applied in the analysis of functional modules in protein-protein networks \cite{alvarez2006large,wuchty2005peeling}.
Here, we compare the $k$-core stability of Human and Yeast (Figs. \ref{fig:SV_k_b_human}, \ref{fig:SV_k_b_yeast}). Human is shown to be more stable, as can be inferred from the range of values in the profile---$1\%$ to $35\%$ for Human and $3.4\%$ to $100\%$ for Yeast. Moreover, the profile for Human is smoother than Yeast. These results confirm our intuition that proteins have a more complex functional structure in Humans compared to other organisms \cite{Zitnik454033}. We also show similar results for clustering coefficient and efficiency, which are other popular robustness measures for networks \cite{ellens2013graph}, within the same core set of vertices to facilitate the comparison. Both competing metrics fail to effectively assess robustness for varying values of $k$ and budget. In particular, the clustering coefficient of the networks remain mostly unchanged after edge deletions. The effect of network efficiency minimization over the core of the network does not necessarily increase with the budget, which is counter-intuitive. More specifically, efficiency minimization often fails to break dense substructures of the network, even for large values of budget.

\begin{figure}[t]
    \centering
     \subfloat[Human (Core resilience)]{\includegraphics[width=0.2\textwidth,height=0.16\textwidth]{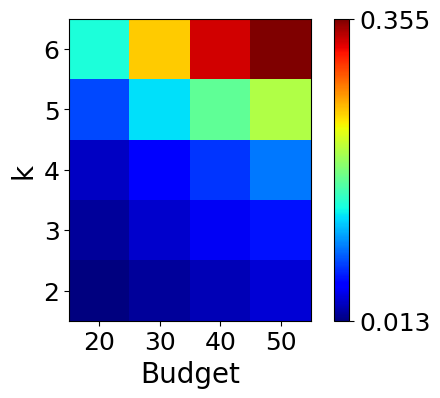}\label{fig:SV_k_b_human}}
    \hspace{1mm}
    \subfloat[Yeast (Core resilience)]{\includegraphics[width=0.2\textwidth,height=0.16\textwidth]{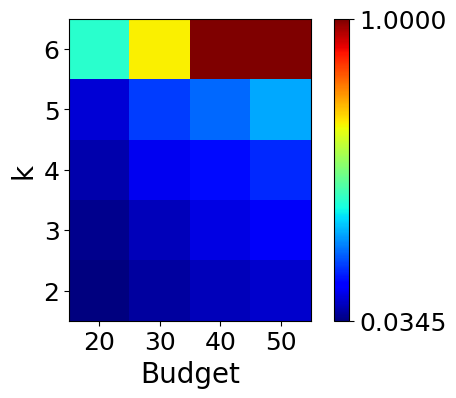}\label{fig:SV_k_b_yeast}}\\
    \vspace{-3mm}
    \subfloat[Human (Clustering coefficient)]{\includegraphics[width=0.2\textwidth,height=0.16\textwidth]{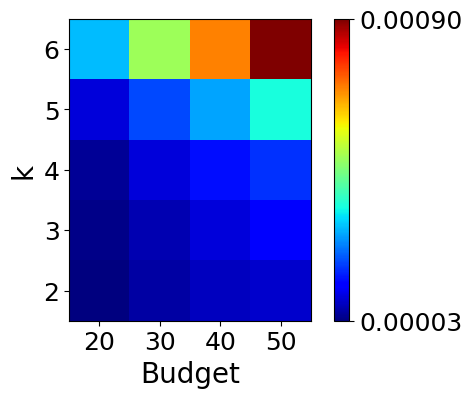}\label{fig:SV_k_b_human_cc}}
    \hspace{1mm}
    \subfloat[Yeast (Clustering coefficient)]{\includegraphics[width=0.2\textwidth,height=0.16\textwidth]{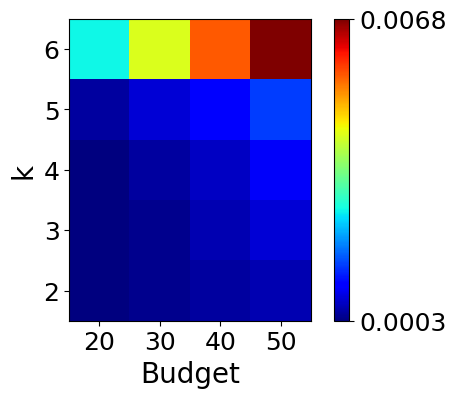}\label{fig:SV_k_b_yeast_cc}}\\
      \vspace{-3mm}
    \subfloat[Human (Efficiency)]{\includegraphics[width=0.2\textwidth,height=0.16\textwidth]{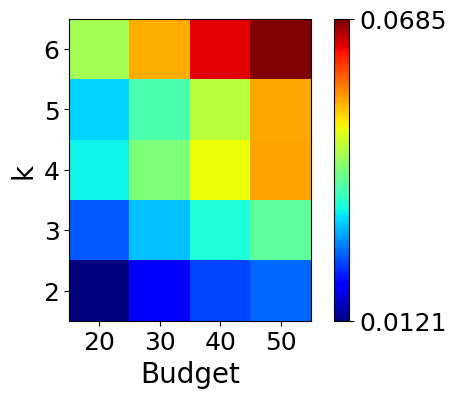}\label{fig:SV_k_b_human_sp}}
    \hspace{1mm}
    \subfloat[Yeast (Efficiency)]{\includegraphics[width=0.2\textwidth,height=0.16\textwidth]{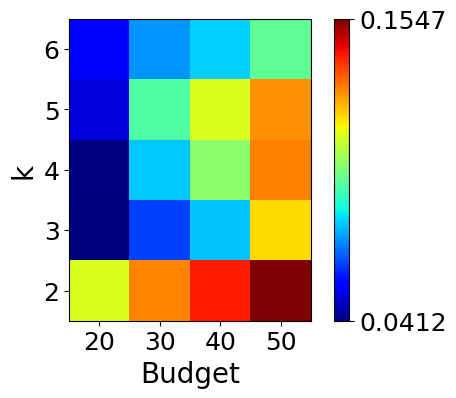}\label{fig:SV_k_b_yeast_sp}}
    \vspace{-3mm}
    \caption{Core resilience (a, b) and other robustness metrics, clustering coefficient (c, d) and efficiency (e, f) \cite{ellens2013graph}, for the Human and Yeast protein-protein interaction networks. \label{fig:SV_human_yeast}}
    
\end{figure}

\end{document}